\documentclass{lmcs}
\pdfoutput=1
\usepackage[utf8]{inputenc}

\usepackage{lastpage}
\lmcsdoi{19}{2}{16}
\lmcsheading{}{\pageref{LastPage}}{}{}%
{Nov.~29,~2022}{Jun.~07,~2023}{}


\usepackage{xspace}
\usepackage{mathpartir}
\mprset{vskip=3pt}
\usepackage{mathtools}
\usepackage{color}
\usepackage{tikz}
\usetikzlibrary{calc,trees,positioning,arrows,chains,shapes.geometric,%
	decorations.pathreplacing,decorations.pathmorphing,shapes,%
	matrix,shapes.symbols}

\tikzset{
	>=stealth',
	punktchain/.style={
		rectangle, 
		rounded corners, 
		draw=black, very thick,
		text width=15.5em, 
		minimum height=2em, 
		text centered, 
		on chain},
	line/.style={draw, thick, <-},
	element/.style={
		tape,
		top color=white,
		bottom color=blue!50!black!60!,
		minimum width=8em,
		draw=blue!40!black!90, very thick,
		text width=10em, 
		minimum height=3.5em, 
		text centered, 
		on chain},
	every join/.style={->, thick,shorten >=1pt},
	decoration={brace},
	tuborg/.style={decorate},
	tubnode/.style={midway, right=2pt},
}

\usepackage{xcolor}
\newcommand\hl[1]{\textcolor{black}{#1}}

\newtheorem{Rule}[thm]{Rule}
\newtheorem{definition}[thm]{Definition}
\newtheorem{lemma}[thm]{Lemma}
\newtheorem{theorem}[thm]{Theorem}

\usepackage[most]{tcolorbox}
\newtcbtheorem[use counter=thm, number within=section]{definitionbox}{Definition}{
  colback=white,
}{def}
\makeatletter
\newcommand\tcb@cnt@definitionboxautorefname{Definition}
\makeatother

\usepackage{graphicx}

\definecolor{pblue}{rgb}{0.13,0.13,1}
\definecolor{pgreen}{rgb}{0,0.5,0}
\definecolor{pred}{rgb}{0.9,0,0}
\definecolor{pgrey}{rgb}{0.46,0.45,0.48}
\definecolor{javapurple}{rgb}{0.5,0,0.35}

\usepackage{listings}

\newlength{\MaxSizeOfLineNumbers}
\newcommand\Q\lstinline
\newcommand{\hoare}[3]{\mathtt{\{#1\} \; #2 \; \{#3\}}}
\newcommand{\E}{\ensuremath{\mathcal{E}}}
\newcommand{\TD}{\ensuremath{\mathit{TD}}}
\newcommand{\CD}{\ensuremath{\mathit{CD}}}
\newcommand{\Body}{\ensuremath{\mathit{Body}}}
\newcommand{\MH}{\ensuremath{\mathit{MH}}}
\newcommand{\vs}{\ensuremath{\mathit{vs}}}
\newcommand{\es}{\ensuremath{\mathit{es}}}
\newcommand{\Cs}{\ensuremath{\mathit{Cs}}}
\newcommand{\Ds}{\ensuremath{\mathit{Ds}}}
\newcommand{\Ms}{\ensuremath{\mathit{Ms}}}
\newcommand{\Name}{\ensuremath{\mathit{Name}}}
\newcommand{\OK}{\ensuremath{\mathit{OK}}}
\newcommand{\Pre}{\ensuremath{\mathit{Pre}}}
\newcommand{\Post}{\ensuremath{\mathit{Post}}}

\newcommand{\corc}{\textsc{CorC}\xspace}
\newcommand{\CbCBlock}{\textsc{CbC-Block}\xspace}
\newcommand{\TraitCbC}{\textsc{TraitCbC}\xspace}

\newcommand{\newKW}{\ensuremath{\mathtt{new}}}
\newcommand{\methodKW}{\ensuremath{\mathtt{method}}}
\newcommand{\result}{\ensuremath{\mathtt{result}}}
\newcommand{\this}{\ensuremath{\mathtt{this}}}
\newcommand{\interface}{\ensuremath{\mathtt{interface}}}

\usepackage{multirow}
\usepackage{multicol}
\usepackage{booktabs}
\usepackage{tabularx}
\newcolumntype{b}{X}
\newcolumntype{n}{>{\hsize=.4\hsize}X}
\newcolumntype{s}{>{\hsize=.2\hsize}X}
\newcolumntype{Y}{>{\centering\arraybackslash}X}
\makeatletter
\newcommand\cellwidth{\TX@col@width}
\makeatother
\makeatletter
\let\tx@\TX@endtabularx
\def\restoretx{\let\TX@endtabularx\tx@}
\makeatother

\usepackage{ltablex}

\begin{document}
\renewcommand{\thelstlisting}{\arabic{lstlisting}}

\title{Flexible Correct-by-Construction Programming}
%\titlerunning{Traits for Correct-by-Construction Programming}

\author{Tobias Runge\lmcsorcid{0000-0002-9154-7743}}[a,b]
\author{Tabea Bordis\lmcsorcid{0009-0003-2886-0862}}[a,b]
\author{Alex Potanin\lmcsorcid{0000-0002-4242-2725}}[d]
\author{Thomas Thüm\lmcsorcid{0000-0001-8069-9584}}[e]
\author{Ina Schaefer\lmcsorcid{0000-0002-7153-761X}}[a,b,c]

\address{Institute of Information Security and Dependability (KASTEL), Karlsruhe Institute of Technology, Germany}
\email{\{tobias.runge, tabea.bordis, ina.schaefer\}@kit.edu}

\address{Institute of Software Engineering and Automotive Informatics, TU Braunschweig, Germany}

\address{School for Data-Science and Computational Thinking, Stellenbosch University, South Africa}

\address{School of Computing, Australian National University, Australia}
\email{alex.potanin@anu.edu.au}

\address{Institute of Software Engineering and Programming Languages, University of Ulm, Germany}
\email{thomas.thuem@uni-ulm.de}

%\author{double blind}
%\author{Tobias Runge\inst{1,2} \and Alex Potanin\inst{3} \and Thomas Thüm\inst{4} \and Ina Schaefer\inst{1,2}}
%\institute{TU Braunschweig, Germany
%	\and Karlsruhe Institute of Technology, Germany
%	\and Australian National University, Australia
%	\and University of Ulm, Germany\\
%	\email{\{tobias.runge,ina.schaefer\}@kit.edu, alex.potanin@anu.edu.au,\\ thomas.thuem@uni-ulm.de}
%}
%\authorrunning{Tobias Runge et al.}

\begin{abstract}
Correctness-by-Construction (CbC) is an incremental program construction process to construct functionally correct programs. The programs are constructed stepwise along with a specification that is inherently guaranteed to be satisfied. CbC is complex to use without specialized tool support, since it needs a set of predefined refinement rules of fixed granularity which are additional rules on top of the programming language. Each refinement rule introduces a specific programming statement and developers cannot depart from these rules to construct programs.
CbC allows to develop software in a structured and incremental way to ensure correctness, but the limited flexibility is a disadvantage of CbC.
In this work, we compare classic CbC with \CbCBlock and \TraitCbC. Both approaches \CbCBlock and \TraitCbC, are related to CbC, but they have new language constructs that enable a more flexible software construction approach. We provide for both approaches a programming guideline, which similar to CbC, leads to well-structured programs.
\CbCBlock extends CbC by adding a refinement rule to insert any block of statements.
Therefore, we introduce \CbCBlock as an extension of CbC. \TraitCbC implements correctness-by-construction on the basis of traits with specified methods.
We formally introduce \TraitCbC and prove soundness of the construction strategy.
All three development approaches are qualitatively compared regarding their programming constructs, tool support, and usability to assess which is best suited for certain tasks and developers.
\end{abstract}

\maketitle

\section{Introduction}
\emph{Correctness-by-Construction} (CbC)~\cite{dijkstra-book,gries-book,kourie2012correctness,morgan-book} is a methodology in the field of formal methods to incrementally construct functionally correct programs guided by a pre-/postcondition specification.\footnote{The approach should not be confused with other CbC approaches such as CbyC of Hall and Chapman~\cite{hall2002correctness}. CbyC is a software development process that uses formal modeling techniques and analysis for various stages of development (architectural design, detailed design, code) to detect and eliminate defects as early as possible~\cite{Chapman:2006}. We also exclude data refinement from abstract data types to concrete ones during code generation as for example in Isabelle/HOL~\cite{Haftmann2013}.} In contrast to post-hoc verification, where a program is typically specified and verified after implementing it, CbC is based around successively creating a program together with the specification. This is achieved by applying refinement rules from a small set of defined rules where in each refinement step, an abstract statement (i.e., a hole in the program) is refined to a more concrete implementation that can still contain some nested abstract statements. While refining the program, the correctness of the whole program is guaranteed through applicability conditions that are defined in the refinement rules. The construction ends when no abstract statement is left.

The underlying idea of this specification-first, refinement-based approach is that developers are forced to think about their algorithm more thoroughly rather than having a trial-and-error verification approach. This trial-and-error verification can oftentimes be experienced with post-hoc verification because programs are implemented first and therefore not well-structured for the verification process which leads to tedious verification work. Additionally, through the structured reasoning discipline that is enforced by the refinement rules in CbC, errors are more likely to be detected earlier in the design process, and it is argued that program quality increases and verification effort is reduced~\cite{kourie2012correctness,watson2016correctness}.

%\emph{Correctness-by-Construction} (CbC)~\cite{dijkstra-book,gries-book,kourie2012correctness,morgan-book} is a methodology that incrementally constructs correct programs guided by a pre-/postcondition specification.\footnote{The approach should not be confused with other CbC approaches such as CbyC of Hall and Chapman~\cite{hall_chapman_2002}. CbyC is a software development process that uses formal modeling techniques and analysis for various stages of development (architectural design, detailed design, code) to detect and eliminate defects as early as possible~\cite{Chapman:2006}. We also exclude data refinement from abstract data types to concrete ones during code generation as for example in Isabelle/HOL~\cite{Haftmann2013}.}
%%We also focus on refinement of abstract methods to concrete implementations by using a set of proven refinement rules and
%CbC uses small tractable refinement rules where 
%in each refinement step, an abstract statement (i.e., a hole in the program) is refined to a more concrete implementation that can still contain some nested abstract statements. While refining the program, the correctness of the whole program is guaranteed through the check of conditions in the refinement rules.
%The construction ends when no abstract statement is left.
%Through the structured reasoning discipline that is enforced by the refinement rules, it is claimed that
%program quality increases and verification effort is reduced~\cite{kourie2012correctness,watson2016correctness}.

Despite these benefits, CbC intuitively has a drawback: The flexibility of creating a program is limited to the set of refinement rules and the rigid, rule-based construction process of applying one rule at a time. This is even increased by the granularity of the rules which explicitly only allow to use one language construct at a time (e.g., one assignment to a variable). Additionally, the refinement rules extend the programming language (i.e., refinement rules are an additional linguistic construct to transform programs), and therefore special tool support (e.g., \corc~\cite{runge2019tool,bordis2022recorc}) is necessary to introduce the CbC refinement process to a programming language. As a result, the barrier to construct programs using CbC is large because the approach at first seems unintuitive and requires effort, knowledge, and special tool support.

%Despite these benefits, CbC has a drawback: the refinement rules extend the programming language (i.e., refinement rules are an additional linguistic construct to transform programs). Special tool support~\cite{runge2019tool} is necessary to introduce the CbC refinement process to a programming language.
%Additionally, the predefined rules have a fine granularity so that for every new statement the developer adds to the program, an application of a refinement rule is necessary. The flexibility of the language is limited through the predefined refinement rules because it is not possible to deviate from this rule-based program construction process. For example, it is tedious to introduce consecutive assignment statements. Even if the developer knows that these statements are correct. Here, it would be helpful to add a block of code that can be checked for correctness in one step. The CbC approach should offer a concept to condense a number of refinement steps into one step if the application is easy to prove correct.
%In the current status, the concepts of CbC increase the effort and necessary knowledge of the developer to construct programs.

In this article, we introduce two alternative correctness-by-construction development approaches that relax the inflexible CbC construction approach without losing the benefits of CbC itself. Both introduce more flexible language constructs to create programs which allow to condense construction steps that tackle the complex and strict programming style of CbC.
The goal is to propose a usable CbC apporach that offers reasonable constructs to develop programs correctly. Therefore, we qualitatively discuss our two proposed approaches and the original CbC approach regarding their programming constructs, tool support, and usability to assess their benefits and drawbacks.

First, we present \CbCBlock which adds new refinement rules. This introduction of new refinement rules should not be seen as a further restriction, but as a relaxation of the procedure.
These new refinement rules increase the ways in which programs can be developed as they allow to refine abstract statements to a specified block of code that fulfills its specification. 
This basically means that this block can contain multiple assignments, selections, or loops whereas with classic CbC for each assignment, selection, and loop a new refinement step is needed. 
Initially, a block is just an abstract placeholder, 
but it has a pre-/postcondition specification so that the introduced specification of the block can be checked against the specification of the refined abstract statement. 
In a next step, the block is instantiated by some code, and it is directly proved that this code fulfills its own specification. The idea of the block rules is similar to a method call, but a block can alter local variables in the method under construction. A block of code can contain further blocks which can be subsequently refined. Consequently, any nesting of blocks may occur.
\CbCBlock is implemented as extension of the \corc tool support.

Second, we present \TraitCbC which is a new software development approach that enables correct-by-construction development by method abstraction and method composition without relying on refinement rules and special tool support.
\hl{\TraitCbC uses \emph{traits}~\cite{ducasse2006traits}, which are a flexible object-oriented language construct supporting a rich form of modular code reuse orthogonal to inheritance.
A trait is a set of \emph{concrete} or \emph{abstract} methods (i.e., the method has either a body or has no body).\footnote{Java interfaces with default methods are a good approximation for what has been called trait in the literature}
Traits can be composed into a larger trait or into a class that contains all methods of all composed traits.}
Trait composition exists as a direct concept of the programming language~\cite{ducasse2006traits} instead of being a program transformation concept, such as the CbC refinement rules.
On the basis of traits, \TraitCbC introduces a programming
guideline for an incremental program construction approach
that guarantees that the resulting program is correct by construction. A construction step comprises the development of a method and direct composition with the existing code base to ensure correctness.
\TraitCbC allows the implementation of any method size and complexity as long as the methods are composable with respect to their specification. Even with this flexibility, \TraitCbC keeps the advantages of a structured incremental development approach.

The contribution of our article is to demonstrate and compare the range of possibilities to develop programs correct by construction from strict rule-based CbC to the more relaxed \CbCBlock to \TraitCbC without any refinement rules. 
%By comparing the processes and showing their strengths and weaknesses, developers can decide for the best fitting process in their development scenario.
%All processes have in common that programs are constructed incrementally with correctness guarantees in each step. How the correctness is enforced is differently. CbC uses a rule-based system to refine the program, whereas \TraitCbC composes correct code units into complete programs.
In this article, we introduce \TraitCbC and explain the typing, reduction, and flattening rules. We give a proof that \TraitCbC guarantees to develop programs correct by construction. We also present the \CbCBlock approach with the block refinement rules. 
All approaches are implemented in the \corc~\cite{runge2019tool} tool support.
We compare and discuss classic CbC, \CbCBlock, and \TraitCbC qualitatively to assess their benefits and drawbacks.
This article extends previous work~\cite{runge2022traits} by introducing the typing and reduction rules of \TraitCbC in detail. The soundness proof of \TraitCbC is also presented in this article. \CbCBlock is a new approach that has not been presented before.

\section{Correctness-by-Construction}

Classic correctness-by-construction (CbC)~\cite{dijkstra-book,kourie2012correctness,morgan-book} is an incremental approach to construct programs. 
%Starting with a pre/post-condition specification, an abstract method is refined step wise to a concrete one. 
CbC uses a \textit{Hoare triple} specification $\hoare{P}{S}{Q}$ stating that if the precondition $\mathtt{P}$ holds, and the statement $\mathtt{S}$ is executed, then the statement terminates and postcondition $\mathtt{Q}$ holds.
The CbC refinement process starts with a Hoare triple where the statement $\mathtt{S}$ is abstract. This abstract statement can be seen as a hole in the program that needs to be filled. 
With a set of refinement rules, an abstract statement is replaced by more concrete statements (i.e., statements in the guarded command language~\cite{dijkstra-book} that can contain further abstract statements). The process stops, when all abstract statements are refined to concrete statements so that no holes remain in the program. 
As each refinement rule is sound and each correct application of a refinement rule guarantees to satisfy the starting Hoare triple, the resulting program is correct by construction~\cite{kourie2012correctness}.
%A correct application of a refinement rule is checked by its conditions.
The CbC approach is strictly tied to this set of predefined refinement rules. A developer cannot deviate from this concept.

In Definition \ref{def:refinementRules}, we present the eight refinement rules of CbC by Kourie and Watson~\cite{kourie2012correctness}. 
The concrete program statements are written in the guarded command language~\cite{Dijkstra:1975}. 
To apply a refinement rule, it has to be checked that side conditions of the rule application are satisfied. This is done by pen-and-paper or with specialized tools~\cite{runge2019tool}.
For example, the skip rule introduces an empty statement that does not alter the program state. This refinement is applicable if and only if the precondition $P$ implies the postcondition $Q$. The composition rule splits the Hoare triple $\{P\}S\{Q\}$ into two Hoare triples by using an intermediate condition $M$. This refinement is applicable, if and only if the two new Hoare triples are correct. Of course, the statements $S_1$ and $S_2$ are still abstract and can be further refined.

\begin{figure}[t]
\begin{definitionbox}[label={def:refinementRules}]{Refinement Rules for the Correctness-by-Construction Approach}{}
	Let $P$ be the precondition, $Q$ be the postcondition, and $S$ be an abstract statement. Then, the Hoare triple $\{P\}S\{Q\}$ is \textbf{refinable} to 
	\begin{itemize}
		\item \textbf{Skip:} $\{P\}$\textbf{skip}$\{Q\}$ \emph{iff} $P$ \emph{implies} $Q$
		\item \textbf{Assignment:} $\{P\}$ $\mathtt{x}$ := $E\{Q\}$ \emph{iff} $P$ \emph{implies} $Q$[$\mathtt{x}$ := $E$] 
		\item \textbf{Composition:} $\{P\}S_1;S_2\{Q\}$ \emph{iff} intermediate condition $M$ exists such that $\{P\}S_1\{M\}$ and $\{M\}S_2\{Q\}$ \emph{hold}
		\item \textbf{Selection:} $\{P\}$\textbf{if} $G_1\rightarrow S_1$ \textbf{elseif} $\dots$ $G_n\rightarrow S_n$ \textbf{fi}$\{Q\}$ \emph{iff}  $P$ \emph{implies} $G_1 \vee\dots\vee G_n$ and $\forall i \in \{1\dots n\} : $ $\{P\wedge G_i\}S_i\{Q\}$ \emph{holds}
		\item \textbf{Repetition:} $\{P\}$\textbf{do} $[I,V]$ $G \rightarrow S $ \textbf{od}$\{Q\}$ \emph{iff}  $P$ \emph{implies} $I$ and $I\wedge \neg G$ \emph{implies} $Q$ and $\{I\wedge G\}S\{I\}$ \emph{holds} and $\{I\wedge G \wedge V=V_0\}S\{I\wedge 0\leq V < V_0\}$ \emph{holds}
		\item \textbf{Weaken precondition:} $\{P'\}S\{Q\}$ \emph{iff} $P$ \emph{implies} $P'$
		\item \textbf{Strengthen postcondition:} $\{P\}S\{Q'\}$ \emph{iff} $Q'$ \emph{implies} $Q$
		\item \textbf{Method Call:} $\{P\}$$\mathtt{m}$($a_1,\dots,a_n$) $\rightarrow$ $b\{Q\}$ \emph{iff} method $\{P'\}$$\mathtt{m}$($p_1,\dots,p_n$) $\rightarrow$ $r\{Q'\}$ exists and $P$ \emph{implies} $P'[p_i\setminus a_i]$ and $Q'[\textrm{old}(p_i) \setminus \textrm{old}(a_i),\ r\setminus b]$ \emph{implies} $Q$
	\end{itemize}
	
	\begin{flushright}
		\cite{runge2019tool,kourie2012correctness}
	\end{flushright}
\end{definitionbox}
\end{figure}

\section{\CbCBlock --- CbC With Block Contracts}

In this section, we introduce the \CbCBlock approach that adds two new refinement rules to classic CbC. The new refinement rules increase the ways to construct programs. Therefore, the rigid CbC approach is loosened  while retaining the benefits of a structured program construction approach.
%\CbCBlock breaks with the classic CbC because it introduces a new block rule. 
A block rule refines an abstract statement to a block that is specified with a block contract (i.e., a pre-/postcondition specification for that block)~\cite{ahrendt2016deductive}. The block is a special statement that can be further refined in two ways.
Similar to an abstract statement, any CbC refinement rules can be applied. Additionally, the block can be instantiated by any sequence of concrete statements and further blocks with a block-instantiation rule.
Thus, a block can be used to condense the application of several CbC refinement rules. For example, a block can be instantiated with a while-loop that already contains a concrete body. This instantiation replaces the application of the repetition rule and at least one assignment rule.
We introduce \CbCBlock with a motivating example and present the block rules to introduce and instantiate blocks. In the end of this section, we present \CbCBlock implemented in the CbC tool \corc\footnote{\url{https://github.com/KIT-TVA/CorC}} and evaluate its usability with a user study.

\subsection{Motivating Example}

%We slightly adjust the implementation of the algorithm to better fit the workflow of \CbCBlock. With \TraitCbC, we have an abstraction on method level. With \CbCBlock, we can refine an abstract statement to any block of code inside a method. For example, we can update a local variable by assigning some expression to this variable. A method call cannot alter local variables of the caller unless the return value is assigned.

In this section, we exemplify the \CbCBlock approach by implementing a \Q@maxElement@ algorithm. 
The \Q@maxElement@ algorithm searches the largest element in a list of integers. The list supports a \Q@get@-method which returns the element at the specified position in this list. A \Q@contains@-method checks that the result is a member of the list.
We iterate with a while-loop through the list and use local variables to temporally save the current largest element.
We use Java and JML~\cite{leavens1998jml} as programming and specification language in the example.

In Listing~\ref{listing:CbCBlockstep0}, we start implementing method \Q@maxElement@ that is specified with a pre- and postcondition contract. The precondition states that the list must contain at least one element. The postcondition states that the largest element in the list is returned.
In this example, we start with a program where some CbC refinement rules are already applied, and then, apply the block rules to finish the implementation. %With \CbCBlock, we are not forced to use the classic CbC refinement rules.

%\begin{lstlisting}[style=sifojava]
%/*@ public normal_behavior
%  @ requires list.size() > 0;
%  @ ensures  list.contains(\result) &&
%  @ (\forall int n; list.contains(n) ==> \result >= n);
%  @*/
%  public int maxElement(List list) {
%    int result;
%    
%  /*@ accessible list, result;
%    @ assignable \hasTo(result). \hasTo(j);
%    @
%    @ public normal_behavior
%    @ requires list.size() > 0;
%    @ ensures  list.contains(result) &&
%    @ (\forall int n; list.contains(n) ==> result >= n);
%    @*/
%    { \abstract_statement B1; }
%    
%    return result;
%  }
%\end{lstlisting}
\begin{figure}[ht]
\begin{lstlisting}[caption={Initial program of \texttt{maxElement}},captionpos=b,label={listing:CbCBlockstep0}]
/*@ public normal_behavior
  @ requires list.size() > 0;
  @ ensures list.contains(\result) && 
  @   (\forall int q; q >= 0 && q < list.size(); \result >= list.get(q));
  @*/
  public int maxElement(List list) {
    int i = list.get(0);
    int j = 1;

  //@ Intm: list.size() > 0 && i == list.get(0) && j == 1;

  /*@
    @ normal_behavior
    @ requires list.size() > 0 && i == list.get(0) && j == 1;
    @ ensures list.contains(i) &&
    @   (\forall int q; q >= 0 && q < list.size(); i >= list.get(q));
    @*/
    { \Block B1; }

  //@ Intm: list.contains(i) &&
  //@   (\forall int q; q >= 0 && q < list.size(); i >= list.get(q));

    return i;
  }
\end{lstlisting}
\end{figure}

The program is already split into three parts using the composition refinement rule with two intermediate conditions between them. In the first part, two local variables \Q@i@ and \Q@j@ are introduced with the assignment rule.
The variable \Q@i@ is used to temporally store the largest element. In the beginning, the largest element (up to that point) is the first element in the list. The variable \Q@j@ is our loop variable to iterate through the list. In the third program part, variable \Q@i@ is returned.
We start with this program state to show that \CbCBlock also supports the standard CbC approach, but we will now use the block rules to exemplify the benefits of \CbCBlock.

In the second part between the two intermediate conditions, the block rule of \CbCBlock is applied. The block \Q@B1@ is specified with a block contract in lines 13--16. 
%First, we express the frame in the block contract. We can access the local variables and the list. We also specify that the block can update the local variables \Q@i@ and \Q@j@. 
For the functional behavior, we specify the values of \Q@i@ and \Q@j@, and the size of the list in the precondition. This specification is equal to the preceding intermediate condition.
The postcondition of the block states that \Q@i@ is the largest element. This postcondition meets the intermediate condition after the block. Therefore, we know that the program is correct under the assumption that block \Q@B1@ fulfills its specification. 
In the next steps, we concretize the block, and the applied refinement rule guarantees that the instantiated block is correct according to its specification. We can concretize the block either in one step, by
instantiating the block with concrete Java code or stepwise by instantiating the
block with some Java statements and other blocks.

We decide to partially implement the block. In Listing~\ref{listing:CbCBlockstep1}, we define the block that should be refined by referring to block \Q@B1@ in line 1 and repeat the specification of that block. Inside the curly brackets the instantiation is shown. We implement the block with a while-loop. We iterate through the list as long as variable \Q@j@ is smaller than the size of the list. This is stated in the loop guard. 
The loop is specified with a loop invariant in lines 10--12. Thereby, the variable \Q@i@ stores the largest element of the already checked elements up to the index \Q@j@. The index \Q@j@ is inside the bounds of the list. We use the difference between the size of the list and \Q@j@ as loop variant. As variable \Q@j@ increases in each iteration, the difference decreases, and the loop thus terminates. The increase of \Q@j@ is already implemented at the end of the while-loop. The body of the loop contains another block \Q@B2@ in lines 15--22.
The precondition of the block is the loop invariant with the difference that we know that variable \Q@j@ is smaller than the size of the list.
This block should update variable \Q@i@ that contains the largest element. 
We want to compare the largest element with the next element in the list. If that element is larger, variable \Q@i@ is updated. We checked one more element of the list, and therefore, we increase the range of the universal quantifier in the postcondition.
%Therefore, we specify that all variables are accessible, and that we can assign to variable \Q@i@. 
This instantiation condenses the application of three CbC refinement rules, the repetition rule to create the loop, a composition rule, and an assignment rule for the loop body.

\begin{figure}[ht]
\begin{lstlisting}[caption={Refinement of block \texttt{B1}},captionpos=b,label={listing:CbCBlockstep1}]
Block B1;

/*@ 
  @ normal_behavior
  @ requires list.size() > 0;
  @ ensures list.contains(i) && 
  @   (\forall int q; q >= 0 && q < list.size(); i >= list.get(q));
  @*/
  { 
  //@ loop_invariant list.contains(i) && j > 0 && j <= list.size() &&
  //@   (\forall int q; q >= 0 && q < j; i >= list.get(q));
  //@ decreases list.size() - j;
    while (j < list.size()) {

    /*@ 
      @ normal_behavior
      @ requires list.contains(i) && j > 0 && j < list.size() &&
      @   (\forall int q; q >= 0 && q < j; i >= list.get(q));
      @ ensures list.contains(i) && j > 0 && j < list.size() && 
      @   (\forall int q; q >= 0 && q < j+1; i >= list.get(q));
      @*/
      { \Block B2; }

      j = j + 1;
    }
  }
\end{lstlisting}
\end{figure}

The next step is to verify that the instantiation satisfies the block contract.
Starting with the precondition and after executing the introduced instantiation, the postcondition of the block contract must be fulfilled. The details of checking this instantiation are explained in the next section. 
When the correctness of this instantiation is shown, we can continue to instantiate the next block \Q@B2@.

In Listing~\ref{listing:CbCBlockstep2}, the instantiation of block \Q@B2@ implements the case when a larger element is found. The functional pre- and postcondition of the block differ by the range of considered elements. In the postcondition, the range is increased by one. In this block, we compare the current largest element \Q@i@ with the element at index \Q@j@. If the element at index \Q@j@ is larger, we update variable \Q@i@. In the other case, \Q@i@ is still the largest element and not updated. Again, we condense CbC refinement rules by instantiating the block with concrete code.
We have to verify that the instantiation is correct. If this is done, we have finished the refinement process because no further block or any abstract statement is left.

\begin{figure}[ht]
\begin{lstlisting}[caption={Refinement of block \texttt{B2}},captionpos=b,label={listing:CbCBlockstep2}]
Block B2;

/*@
  @ normal_behavior
  @ requires list.contains(i) && j > 0 && j < list.size() &&
  @   (\forall int q; q >= 0 && q < j; i >= list.get(q));
  @ ensures list.contains(i) && j > 0 && j < list.size() && 
  @   (\forall int q; q >= 0 && q < j+1; i >= list.get(q));
  @*/
  { 
    if (list.get(j) > i) {
      i = list.get(j);
    }
  }
\end{lstlisting}
\end{figure}

By guaranteeing the correctness of all refinement steps, we can conclude that the whole program is correct by construction.
The resulting program is shown in Listing~\ref{listing:CbCBlockfinal}. Here, the blocks are recursively replaced with their instantiation. The specification is limited to the method contract and the loop invariant and variant annotations.
By stepwise refining the program, we can detect errors when proving single refinement steps. 
This locality of information helps to track down errors more easily than with monolithic post-hoc verification.
%However, we are free to use any granularity of refinements. For simple problems, a method could even be implemented in one step.

\begin{figure}[ht]
\begin{lstlisting}[caption={Final implementation of \texttt{maxElement}},captionpos=b,label={listing:CbCBlockfinal}]
/*@ public normal_behavior
  @ requires requires list.size() > 0;
  @ ensures list.contains(\result) &&
  @   (\forall int q; q >= 0 && q < list.size(); \result >= list.get(q));
  @*/
  public int maxElement(List list) {
    int i = list.get(0);
    int j = 1;
  //@ loop_invariant list.contains(i) && j > 0 && j <= list.size() && 
  //@   (\forall int q; q >= 0 && q < j; i >= list.get(q));
  //@ decreases list.size() - j;
    while (j < list.size()) {
      if (list.get(j) > i) {
        i = list.get(j);
      }
      j = j + 1;
    }
    return i;
  }
\end{lstlisting}
\end{figure}

\subsection{Block Refinement Rules of \CbCBlock}

In this section, we describe how refinement rules are added to establish the \CbCBlock approach. We describe the refinement rule to introduce a block and the refinement rule to instantiate a block with concrete code.

\hl{For the block rules to be syntactically applicable, we extend Java to write a block with a name. Normally, a block in Java is just a sequence of Java statements inside curly brackets. In addition, Ahrendt et al.\cite{ahrendt2016deductive} defined block contracts to specify the behavior of a Java block similar to a method~\cite{meyer1992applying,leino1995toward}.
To establish a CbC refinement process, we introduce a specified block as an abstract statement in \CbCBlock with an according refinement rule.
In the refinement rule, we use the Hoare triple notation that is also used for the classic CbC refinement rules.
We focus on functional pre-/postconditions and exclude regular and irregular termination of blocks for \CbCBlock.
For the instantiation of a block (e.g., to write a sequence of Java statements that fulfill the specification), we follow the syntax of a concrete block in Java, but we add a name for reference.}

An abstract statement is refined by the block-introduction rule to a block with a name and a block contract. Thus, a block name is an abstract placeholder. The side condition of the refinement rule guarantees the correctness of the program to be developed. For the block-introduction rule, we have to check three parts. First, the precondition of the refined abstract statement must imply the precondition of the block. This ensures that the pre-state of the block is satisfied, and the block can be executed. Second, the postcondition of the block must imply the postcondition of the refined abstract statement to continue the program after the block. Third, the block must satisfy its own contract.
As the block can be seen as a Hoare triple, any CbC refinement rule can be applied to the block.
The check of the side condition of the applied refinement rule guarantees the correctness of the block under development.

\begin{Rule}[Block-Introduction]
	Hoare triple $\{P\}S\{Q\}$ is \textbf{refinable} to $\{P'\}$ \Q@Block B@ $\{Q'\}$\\ \emph{iff} $P$ \emph{implies} $P'$ and $Q'$ \emph{implies} $Q$ and $\{P'\}$ \Q@Block B@ $\{Q'\}$ \emph{holds}.
\end{Rule}

With the block-instantiation rule, we allow to instantiate a block with concrete code that can contain further blocks (see the instantiation in Listing~\ref{listing:CbCBlockstep1}).
For application, it must be checked that this instantiation fulfills the block contract. 
We use the capabilities of program verification.
We translate the block to a method and verify whether this translated \emph{block-method} fulfills its contract.
Thus, we have to prove that the dynamic formula $P \rightarrow$ \Q@<statement;...>@$Q$ is fulfilled. Assuming the precondition, the postcondition must be satisfied after executing the statements in the block.
Dynamic logic extends first-order logic with two operators. A diamond modality \Q@<p>@$Q$ and a box modality \Q@[p]@$Q$ with a program \Q@p@ and a dynamic logic formula $Q$. Intuitively, the diamond modality
states total correctness of the program, and partial correctness is stated with the box modality.

The translation from a block to a block-method is as follows. The block contract is translated to the contract of the block-method.
The translated block-method is added to the same class as the method in which the block is declared.
The statements within the block become the body of the block-method. As block could introduce local variables that are already declared in the surrounding method~\cite{kohlbecker1986hygienic}, an $\alpha$-conversion~\cite{barendregt1984lambda} is necessary to safely rename identifiers.
A block does not have the same scope of a complete method and neither has parameters nor a return type. Declarations of parameters and local variables have to be added to the block-method, so that is has the same scope as the method.
Therefore, we translate accessible variables of the block to parameters of the block-method, and assignable variables of the block to fields of the class containing the block-method. 
This differentiation is done because a contract can only access parameter values before execution of the method, but it can access the modified values of fields.
Accessible or assignable fields of the class are usable because the block-method is added to the class for verification purposes.
The return type of the block-method is \Q@void@ because we exclude the use of return statements inside the block.
This transformation is limited in its expressiveness as we are excluding irregular termination, but sufficient to demonstrate the correctness-by-construction process for normal execution.

\begin{Rule}[Block-Instantiation]
	Hoare triple $\{P\}$ \Q@Block B@ $\{Q\}$ is \textbf{refinable} to\\ $\{P\}$ \Q@<statement;...>@ $\{Q\}$ \emph{iff} $P \rightarrow$ \Q@<statement;...>@ $Q$, where \Q@<statement;...>@ is any sequence of concrete program statements possibly containing further blocks.
\end{Rule}

\subsection{Discussion}

In this subsection, we discuss the block refinement rules in comparison to related approaches that allow to introduce code sequences, such as method calls, macro expansions, and abstract execution.

\paragraph{Difference to the Method Call Rule.}
The difference between the block refinement rules and the method call refinement rule is that for a method call only the contract is used to verify correctness of the caller. The content of the method is assumed to be correct with respect to the method's contract. 
With the block rules, both the contract and the content of the block are always checked for correctness. A big difference between the block rules and the method call rule is their scope. In a method, only variables of the method are changed and no local variables of the calling context. A block allows the modification of local variables as demonstrated in the motivating example.

\paragraph{Difference to Macro Expansion.}
Macro expansion is a textual transformation of input source code. A preprocessor replaces macros with concrete source code. This is similar to the block-instantiation rule, where a block name is replaced with concrete source code.
As for our block-instantiation rule, a macro expansion can capture identifiers already used in the surrounding scope. Therefore, hygienic macro expansion uses $\alpha$-conversion~\cite{barendregt1984lambda} to rename identifiers. 
The difference to \CbCBlock is that \CbCBlock demands a specification for a block that is introduced in the block-introduction rule. Additionally, the block-instantiation rule starts a procedure that verifies whether the block instantiation fulfills its specification. A macro expansion is just a transformation of code.

\paragraph{Abstract Execution for Correctness-by-Construction.}

Abstract execution~\cite{steinhofel2019abstract} is a technique to specify and verify programs with partially abstract parts. Abstract execution generalizes symbolic execution. It is tailored to Java, but the principles are applicable to other sequential languages. Java and JML are extended with the concept of abstract program element (APE); an abstract statement or an abstract expression. An APE is a placeholder for any program part with or without side effects. To verify the correctness of programs containing abstract program elements, these elements are specified with a contract similar to a block contract.
The extended specification language of abstract execution allows to specify the behavior of the program element in cases of regular or irregular termination including side effects~\cite{steinhofel2019abstract}.
The strength of abstract execution is the reasoning of irregular termination that we exclude in \CbCBlock.

APEs can be used similar to blocks of \CbCBlock to establish a process for refinement-based program construction. With abstract execution, we write programs containing APEs. These programs  can be verified to be correct under the assumption that the APEs fulfill their specifications. In a refinement step, an APE is replaced by a program part that contains concrete statements and possibly other abstract program elements. We have to verify that the insertion fulfills the specification of the refined APE. This refinement is repeated until no APE remains. 
Similar to classic CbC, this process does not require a program to be monolithically verified, but it is sufficient that each APE replacement is verified to conclude that the program is correct by construction. 
This process is the same as for \CbCBlock if we always instantiate a block without using any other CbC refinement rule. We still argue that the application of other CbC refinement in tandem with blocks is beneficial because they enforce a structured program construction process where developers think about the implementation more thoroughly.
Therefore, we decided for \CbCBlock as presented instead of utilizing abstract execution because \CbCBlock is the sweet spot between expressiveness and changes to the program construction process of classic CbC. Combining classic CbC refinement rules with abstract execution requires major changes to classic CbC, so that the strength of abstract execution is usable (i.e., the refinement rules must be adapted to consider irregular termination).

\subsection{Implementation}

In this subsection, we describe the implemented tool support for \CbCBlock. 
Classic CbC is already supported by the \corc tool~\cite{runge2019tool}.
\corc has a graphical and a textual editor to develop programs.
In this work, we extend \corc with the new rules of \CbCBlock. 
%We start with the base framework of the IDE and describe how the IDE implements the mandatory features to support \CbCBlock.
\hl{The textual IDE is implemented in Eclipse with Xtext.\footnote{\url{https://www.eclipse.org/Xtext/}} Xtext provides the functionality to develop IDEs for domain-specific languages. We use Xtext to establish an editor for \corc-programs that consist of JML, Java, 
and the CbC-specific keywords. The grammar of a \corc-program, which represents a refinement-based construction according to CbC, is defined in Xtext.
Based on this grammar, Xtext supports syntax checks, highlighting and auto completion.}

\hl{In \corc, we implemented the refinement rules of Definition~\ref{def:refinementRules}. 
For \CbCBlock, we added refinement rules of block-introduction and  block-instantiation.
An instantiation is written in a separated program starting with the name of the refined block and followed by the contract and the block's implementation.}
To verify that a block-instantiation fulfills its block contract, a generator is implemented. It transforms a block-instantiation to a method and starts the verification process by calling KeY~\cite{ahrendt2016deductive}. The transformation to a method follows the concept presented for the block-instantiation rule before.
The generator also creates the final method implementation if all refinement steps are proven. All block instantiations must be recursively inserted to get the final method implementation. The final method implementation can be integrated into an existing code base. This generator can also construct partial methods when some parts are not fully refined. This is helpful in intermediate steps of the construction to retain an overview of the current method.

%With \CbCBlock, the program under construction can be split in a root method and several block instantiations. To keep track of all refinements, an overview is helpful that visualizes the introduced blocks and their instantiations in a tree structure. It is a tree because an instantiation can introduce zero to many blocks that are again instantiated. Circular instantiations are prohibited and could be detected with this overview. The developer also sees which instantiations are already verified. If an instantiations does not fulfill the block contract, the overview indicates this open proof obligation.
%Additionally, the developer sees which blocks are not instantiated yet.
%The overview is also helpful to navigate between the different files.

\subsection{Evaluation with a User Study}

We evaluate \CbCBlock with a user study. We compare \CbCBlock with classic CbC by (dis)allowing the use of the block rules to answer the following research question. 
\begin{description}
	\item[RQ1] Does the \CbCBlock approach improve classic CbC in terms of usability?
\end{description}

\hl{In the user study, we engaged five participants that know classic CbC and the \corc tool. Knowledge of classic CbC and \corc is a necessary prerequisite because a new feature was evaluated that could only be understood if the participants already knew the classic CbC approach in \corc. Then, they can estimate the benefits of the new block refinement rules.
With this small number of participants, it was possible that everyone could solve the study tasks consecutively.}

Each participant had to implement two algorithms, 
one algorithm with \CbCBlock and the block rules and one algorithm with classic CbC and without the block rules. 
\hl{The algorithms are \Q@maxElement@ and \Q@dutchFlag@. The \Q@maxElement@ algorithm was already introduced as the motivating example. The \Q@dutchFlag@ algorithm sorts a list containing only three different elements. In the original description, each element has either a red, a blue or a white color. The elements of the list are to be reordered so that the list results in the national flag of the Netherlands (red, white, blue). We adapted the task to a list that contains an unknown quantity of the numbers 0, 1, and 2.
Both algorithms can be implemented in a few lines of code with one loop through the list of elements. As both algorithms are explained to the participants, we expected that the correct implementation is possible without major problems.}
For each task, they had 30 minutes.
We split the participants in two groups using the Latin Square design~\cite{wohlin2012experimentation}. Each group implemented the algorithms in the same order, but the approaches were used crosswise to address possible learning effects through an order of the tools. After implementing both algorithms, we conducted a structured interview. 
The questions of the interview are presented in Appendix~\ref{appendix:questions}.

We summarize the most important answers of the participants and discuss the findings. 
The tasks were correctly solved by all participants. In general, the participants needed more time for the task with \CbCBlock than for the task with classic CbC. As we only had five participants, these statistics are only of limited significance.
The participants followed the \CbCBlock approach and refined a program stepwise using the block rule. All participants considered the introduction of the block rules to be useful. The rules save the application of other CbC refinement rules during construction.
For \CbCBlock, the participants positively mentioned the familiarity with a textual editor, the grouping of statements for one refinement step by using a block, and the freedom to not be bound to the classic CbC refinement rules.
For classic CbC, the answers are in line with previous user studies~\cite{runge2019user,runge2021user}. The participants positively mentioned the visual overview of the refinements and the status of the verification. They liked the fine-grained feedback for every applied refinement rule. As \CbCBlock extends CbC, these positive answers also apply for \CbCBlock. The participants stated for both approaches that the incremental construction helps to track down errors. The participants still miss more assistance if a proof cannot be closed.

While we observed the participants, we noticed that
both approaches need a correct and sufficient specification as a starting point. If that is the case, refining and checking side-conditions can be very successful. If, on the contrary, the specification needs to be adjusted in the process, the effort to verify the program increases drastically.
With classic CbC, the participants are forced into the process of refining and verifying top-down. With \CbCBlock, the participants have more freedom to develop the program. 

Regarding the finished tasks, all participants had experience with \corc. Therefore, it is not surprising that all participants finished the task. They never used \CbCBlock before, but the participants conceptually understood the features of \CbCBlock and accepted the expansion well.  
Nonetheless, the participants need more time to fully understand the IDE and the programming workflow of \CbCBlock. This results in a longer time to implement the algorithms.

We can answer \textbf{RQ1} that \CbCBlock is a promising feature to increase the usability of \corc, but as for each new feature, developers need time to get used to. Some answers of the participants highlight that with more training and better tool support, they are willing to use the \CbCBlock approach to construct correctness-critical programs.

\paragraph{Threats to Validity}
\hl{In our user study, we had only five participants. Due to the small number of participants, the qualitative results that we collected in the structured interview are not generalizable. Nevertheless, the results are relevant, since the users are experts in \corc and can therefore assess the advantages of the extension. A more comprehensive evaluation with non-experts is not possible because they cannot properly interact with the tool.
The participants only implemented and verified two small algorithms in our experiment, and therefore, we cannot generalize the results to larger problems.}

\section{\TraitCbC}

In this section, we introduce \TraitCbC with a motivational example. We present\linebreak \TraitCbC formally and prove soundness of the TratiCbC program construction approach. In the end of this section, we show the proof-of-concept implementation and a feasibility evaluation.

\TraitCbC uses method abstraction and method composition to enable an incremental CbC-based development approach. This approach on method level allows a flexible way to construct the desired program with any number and size of auxiliary methods.
A developer starts by implementing a method (e.g., a method \texttt{a}) in a first trait. Similar to classic CbC, the method can contain holes that are refined in subsequent steps. A hole in \TraitCbC is an abstract method (e.g., an abstract method \texttt{b}) that is called in method \texttt{a}; that is, a call to an abstract method corresponds to an abstract statement in classic CbC.
In the next step, one of these new abstract methods (e.g., \texttt{b}) is implemented in a second trait, again more abstract methods can be declared for the implementation. To be correct, it must be proven that each implemented method satisfies its specifications. Afterwards, the traits are composed; the composition operation checks that the specification of the concrete method \texttt{b} in the second trait fulfills the specification of the abstract method \texttt{b} in the first trait. This incremental program construction approach stops when the last abstract method is implemented, and all traits are composed.

\subsection{Motivating Example}
\label{sec:example}
%\TraitCbC uses method calls and trait composition to incrementally construct correct programs.
%In this section, we go through an example of how our development process enables CbC using traits.
%\TraitCbC starts from a single trait with an abstract method and its specification.
%Then, we implement this method in a second trait with the use of further abstract methods to represent incomplete parts.
%That is, every refinement step introduces a method implementation satisfying its specification, and introduces new abstract methods with any additional new specifications.
%The existing specification guides the program towards a new implementation and may refine or introduce some specifications.
%Then, both traits are composed. In this step we prove that the concrete method fulfills the specification of the abstract method declaration.
%When new abstract methods are introduced, these steps are repeated until all methods are implemented.

%\begin{figure}[t]
%	\centering
%	\includegraphics[width=\linewidth]{graphics/Procedure}
%	\caption{Procedure to refine abstract methods}
%	\label{fig:procedure}
%\end{figure}

%\paragraph{Incremental Construction of MaxElement}
We illustrate using an example of how \TraitCbC enables CbC using traits.
We use an object-oriented language in the code examples.
In Listing~\ref{listing:CbC-Traitstep0}, we construct a method  \Q@maxElement@ that finds the maximum element in a list of numbers.
We slightly adjust the implementation of the algorithm to better fit for \TraitCbC.
With \TraitCbC, we have an abstraction on method level. We utilize methods to outsource program pieces that can be reused (i.e., we want to implement methods that are verified once, but called several times in a program to reduce verification effort).

In this \Q@maxElement@ example, a list has a head and a tail.
Only non-empty lists have a maximum element. This is explicit in the precondition of our specification, where we require that the list has at least one element. In the postcondition, we specify that the result is in the list and larger than or equal to every other element.
In the first step, we create a trait \Q@MaxETrait1@ that defines the abstract method \Q@maxElement@. 
The method \Q@maxElement@ is abstract, i.e., equivalent to an abstract statement in CbC.

\begin{figure}[ht]
\begin{lstlisting}[caption={Initial trait for maxElement},captionpos=b,label={listing:CbC-Traitstep0}]
trait MaxETrait1 {
  @Pre: list.size() > 0
  @Post: list.contains(result) &
    (forall Num n: list.contains(n) ==> result >= n)
  abstract Num maxElement(List list);
}
\end{lstlisting}
\end{figure}

In the second step in trait \Q@MaxETrait2@ in Listing~\ref{listing:CbC-Traitstep1}, we implement the method \Q@maxElement@ using two abstract methods. We introduce an \Q@if-elseif-else@-expression where the branches invoke abstract methods. The guards check whether the list has only one element or whether the current element is larger than or equal to the maximum of the rest of the list. The abstract method \Q@accessHead@ returns the current element, and the abstract method \Q@maxTail@ returns the maximum in the remaining list. So, we recursively search the list for the largest element by comparing the maximum element of the list tail with the current element until we reach the end of the list.

\begin{figure}[ht]
\begin{lstlisting}[caption={Implementation of maxElement with auxiliary methods},captionpos=b,label={listing:CbC-Traitstep1}]
trait MaxETrait2 {
  @Pre: list.size() > 0
  @Post: list.contains(result) &
    (forall Num n: list.contains(n) ==> result >= n)
  Num maxElement(List list) =
    if (list.size() == 1) {accessHead(list)}
    elseif (accessHead(list) >= maxTail(list)) 
      {accessHead(list)}
    else {maxTail(list)}

  @Pre: list.size() > 0
  @Post: result == list.element()
  abstract Num accessHead(List list);

  @Pre: list.size() > 1
  @Post: list.tail().contains(result) &
    (forall Num n: list.tail().contains(n) ==> result >= n)
  abstract Num maxTail(List list);
}
\end{lstlisting}
\end{figure}

%However, \TraitCbC is parametric on the logic, so of course this only holds if such logic is consistent when composed in this manner.
%In particular, the logic needs to somehow take into account ill-founded specifications and infinite recursion.
%We discuss more about the difficulties of handling those cases in the Appendix~\ref{sec:circularity}.

%Here, we simulate a CbC refinement rule that introduce a conditional statement. 
%As in CbC, our  condition to prove is that the then-branch with the guard and the else-branch with the negated guard both imply the postcondition of the \Q@maxElement@ method (cf. the example in Section~\ref{sec:cbc}).

The correct implementation of the method \Q@maxElement@ can be guaranteed under the assumptions that all introduced abstract methods are correctly implemented.
Similar to post-hoc verification, a program verifier conducts a proof of method \Q@maxElement@ and uses the introduced specifications of the methods \Q@accessHead@ and \Q@maxTail@. If the proof succeeds, we know that the first method is correctly implemented. In our incremental CbCTrait approach, we verify each method implementation directly after construction; and so we are able to reuse each implemented method in the following steps (e.g., by calling the method in the body of other methods).
%This facilitates the construction of correct programs.
%This is an advantage in comparison to CbC. In CbC, the verification of a refinement can contain holes in the proofs as long as abstract statements are used. Thus, the proof is propagated until all holes are closed by refining the program. %In our approach, we specify the abstract methods such that a direct verification is possible.

%The standard CbC process can lead to errors when starting with a wrong specification.
%The standard CbC process can lead to unexpected behaviour when starting with a misaligned specification.

We now compose the developed traits to complete the first construction step. To perform the composition \Q@MaxETrait1@ + \Q@MaxETrait2@, we check that the specification of the method \Q@maxElement@ fulfills the specification of the abstract method in the first trait (cf. Liskov substitution principle~\cite{liskov}). In this case, this means checking that:\\
\Q@MaxETrait1.maxElement(..).pre ==> MaxETrait2.maxElement(..).pre@
as well as:\\ 
\Q@MaxETrait2.maxElement(..).post ==> MaxETrait1.maxElement(..).post@.\\
When the composition of two verified traits is successful, the result is also a verified trait.
Note that the composed trait does not need to be verified directly by a program verifier in \TraitCbC because it is correct by construction.
%If there are two verified traits, the composition of the traits is a verified trait.
In this example, the specifications are the same, thus checking for a successful composition is trivial, but this is not generally the case. 
In particular, the logic needs to take into account ill-founded specifications and recursion in the specification.
We discuss more about the difficulties of handling those cases in previous work~\cite{runge2022traits}.

The methods \Q@accessHead@ and \Q@maxTail@ are implemented in the next two construction steps in traits \Q@MaxETrait3@ and \Q@MaxETrait4@\footnote{The methods could also be implemented in one trait.}.
The implementations are shown in Listing~\ref{listing:CbC-Traitstep2} and in Listing~\ref{listing:CbC-Traitstep3}.
As we implement a recursive method, the method \Q@maxTail@ calls the \Q@maxElement@ method, thus \Q@maxElement@ is introduced as an abstract method in this trait. 
%That all implementations comply with their specification can be shown by verifying that the method bodies satisfy the specification.
We have to verify that the method \Q@accessHead@ satisfies its specification using a program verifier.
%To verify that this call is correct, the system creates a proof goal that checks if the precondition of \Q@element@ is satisfied, and also if the postcondition of \Q@accessHead@ is satisfied after executing \Q@element@.
%Depending on the specific kind of verification logic used, some extra checks may be needed to verify the absence of infinite recursion and ill-founded goals.
Similarly, we have to verify the correctness of the method \Q@maxTail@.

\begin{figure}[ht]
\begin{lstlisting}[caption={Implementation of accessHead},captionpos=b,label={listing:CbC-Traitstep2}]
trait MaxETrait3 {
  @Pre: list.size() > 0
  @Post: result == list.element()
  Num accessHead(List list) = list.element()
}
\end{lstlisting}
\end{figure}

\begin{figure}[ht]
\begin{lstlisting}[caption={Implementation of maxTail},captionpos=b,label={listing:CbC-Traitstep3}]
trait MaxETrait4 {
  @Pre: list.size() > 1
  @Post: list.tail().contains(result) &
    (forall Num n: list.tail().contains(n) ==> result >= n)
  Num maxTail(List list) = maxElement(list.tail())

  @Pre: list.size() > 0
  @Post: list.contains(result) &
    (forall Num n: list.contains(n) ==> result >= n)
  abstract Num maxElement(List list);
}
\end{lstlisting}
\end{figure}

As before, all traits are composed, and it is checked that the specifications of the concrete methods fulfill the specifications of the abstract ones.
As we have no contradicting specifications for the same methods, the composition is well-formed. 
%Thus, we showed a CbC process with trait composition to implement correct programs. 
In Listing~\ref{listing:CbC-Traitstep4}, the final program \Q@MaxE@ is shown. All traits are composed. %We also show the class \Q@MaxEFlattened@ where all methods are directly implemented in the class.

\begin{figure}[ht]
\begin{lstlisting}[caption={Trait composition},captionpos=b,label={listing:CbC-Traitstep4}]
class MaxE = MaxETrait1 + MaxETrait2 + MaxETrait3 + MaxETrait4 
\end{lstlisting}
\end{figure}

%class MaxEFlattened {
%	@Pre: list.size() > 0
%	@Post: list.contains(result) &
%	(forall Num n: list.contains(n) ==> result >= n) 
%	Num maxElement(List list) =
%	if (list.size() == 1) {accessHead(list)}
%	elseif (accessHead(list) >= maxTail(list)) 
%	{accessHead(list)}
%	else {maxTail(list)}
%	
%	@Pre: list.size() > 0
%	@Post: result == list.element()
%	Num accessHead(List list) = list.element()
%	
%	@Pre: list.size() > 1
%	@Post: list.tail().contains(result) &
%	(forall Num n: list.tail().contains(n) ==> result >= n)
%	Num maxTail(List list) = maxElement(list.tail())
%}

%class MaxECompleteInlined {
%	@Pre: list.size() > 0
%	@Post: (forall Num n: list.contains(n) 
%	==> result >= n) & list.contains(result)
%	Num maxElement(List list) =
%	if (list.element()>=maxElement(list.tail())) 
%	{list.element()}
%	else {maxElement(list.tail())}
%}

The already proven auxiliary methods in traits can be reused. 
For example, if we want to implement a \Q@minElement@ method as shown in Listing~\ref{listing:CbC-Traitmin},
we could reuse already implemented traits to reduce the programming and verification effort.
The method \Q@minElement@ is implemented in the following in trait \Q@MinE@ with one abstract method. The specification of the method \Q@accessHead@ is the same as for the method \Q@accessHead@ above, so \Q@MaxETrait3@ can be reused.
In this example, we show the flexible granularity of \TraitCbC by directly implementing the else branch, instead of introducing an auxiliary method as for \Q@maxElement@.

\begin{figure}[ht]
\begin{lstlisting}[caption={Implementation of minElement with auxiliary method accessHead},captionpos=b,label={listing:CbC-Traitmin}]
trait MinE {
  @Pre: list.size() > 0
  @Post: list.contains(result) &
    (forall Num n: list.contains(n) ==> result <= n)
  Num minElement(List list) =
    if (list.size() == 1) {accessHead(list)}
    elseif (accessHead(list) <= minElement(list.tail()))
      {accessHead(list)}
    else {minElement(list.tail())}
  
  @Pre: list.size() > 0
  @Post: result == list.element()
  abstract Num accessHead(List list);
}
\end{lstlisting}
\end{figure}

The correctness of \Q@minElement@ is verified with the specifications of the method \Q@accessHead@. %Similar to above, we have to verify that both branches imply the postcondition of \Q@minElement@. 
By composing \Q@MinE@ with \Q@MaxETrait3@, we get a correct implementation of \Q@minElement@.
Note how this verification process supports abstraction:
as long as the contracts are compatible,
methods can be implemented in different styles by different developers to best meet non-functional requirements while preserving the specified observable behavior~\cite{terBeek2018xbc}.
A completely different implementation of \Q@maxElement@ can be used if it fulfills the specification of the abstract method \Q@maxElement@ in trait \Q@MaxETrait1@. This decoupling of specification and corresponding satisfying implementations facilitates an incremental program construction approach where a specified code base is extended with suitable implementations~\cite{damiani2014verifying}.

\subsection{Object-Oriented Trait-Based Language}
In this section, we formally introduce the syntax, type system, reduction, and flattening semantics of a minimal core calculus for \TraitCbC.
We keep this calculus for \TraitCbC parametric in the specification logic so that it can be used with a suitable program verifier and associated logic.
The presented rules to compose traits are conventional. The focus of our work is to enable a CbC approach using traits that developers can easily adopt. Therefore, we present the calculus to prove soundness of \TraitCbC, but focus on the presentation of the advantages of incremental trait-based programming in this paper. Indeed, languages with traits and with a suitable specification language intrinsically enable incremental program construction.

\subsubsection{Syntax}
The concrete syntax of our core calculus for \TraitCbC is shown in Fig.~\ref{f:syntax}, where non-terminals ending with `s' are implicitly defined as a sequence of non-terminals, i.e., $vs ::= v_1 \dots v_n$. We use the metavariables $t$ for trait names, $C$ for class names and $m$ for method names.
A program consists of trait and class definitions. Each definition has a name and a trait expression $\mathit{E}$. The trait expression can be a $\mathit{Body}$, a trait name, a composition of two trait expressions $\mathit{E}$, or a trait expression $\mathit{E}$ where a method is made abstract, written as $E[\mathtt{makeAbstract}\ m]$. A $\mathit{Body}$ has a flag $\mathtt{interface}$ to define an interface, a set of implemented interfaces $\Cs$ and a list of methods $\Ms$. Methods have a method header $\mathit{MH}$ consisting of a specification $S$,  the return type, a method name, and a list of parameters. Methods have an optional method body. In the method body, we have standard expressions, such as variable references, method calls, and object initializations. For simplicity, we exclude updatable state. Field declarations are emulated by method declarations, and field accesses are emulated by method calls.

The specification $\mathit{S}$ in each method header is used to verify that methods are correctly implemented. The specification is written in some logic. In our examples, we will use first-order logic (cf. the example in Section~\ref{sec:example}).
A well-formed program respects the following conditions:

Every $\mathit{Name}$ in $\mathit{Ds}$ must be unique so that $\mathit{Ds}$ can be seen as a map
from names to trait expressions.
Trait expressions $E$ can refer to trait names $t$.
A well-formed $\Ds$ does not have any circular trait definitions like $t = t$
or $t_1 =t_2$ and $t_2 = t_1$.
In a $\Body$, all names of implemented interfaces must be unique and all method names must be unique, so that $\Body$ is a map
from method names to method definitions.
In a method header, parameters must have unique names, and no explicit parameter
can be called $\mathtt{this}$.
%With this specification, we verify that refinements are correctly applied, and during trait composition that concrete methods satisfy the specifications of the abstract methods with the same method headers.

%We give an example in Listing~\ref{code:trait}. The trait \Q@BasicFunction@ implements a simple method \Q@compareElements@ to compare two numbers. The trait \Q@MaximumSearch@ uses an abstract method \Q@compareElements@ and implements a \Q@maximumElement@ method which finds the maximum element in a list. Now, the class \Q@MaximumList@ is implemented by composing both traits. As \Q@compareElements@ is concrete in \Q@BasicFunction@ and abstract in \Q@MaximumSearch@, the composition is conflict free.
%
%%\begin{lstlisting}[caption={Trait example for a bank account},captionpos=b,label={code:trait}]
%%trait BasicAccount = {
%%void update(Account acc, Num x) {
%%acc.setBal(acc.getBal() + x);}
%%}
%%
%%trait TransferAccount = {
%%void transfer(Account acc1, Account acc2, Num x) {
%%update(acc1, -x); update(acc2, x);}
%%abstract void update(Account acc, Num x)
%%}
%%
%%class BankAccount = BasicAccount + TransferAccount
%%\end{lstlisting}
%
%\begin{lstlisting}[caption={Trait composition example},captionpos=b,label={code:trait}]
%trait BasicFunction = {
%  Num compareElements(Num e1, Num e2) = [...]
%}
%
%trait MaximumSearch = {
%void maximumElement() = [...]
%abstract Num compareElements(Num e1, Num e2);
%}
%
%class MaximumList = BasicFunction + MaximumSearch
%\end{lstlisting}
 
\begin{figure}[t]
	\centering
	\setlength{\fboxsep}{3pt}
	\setlength{\fboxrule}{1pt}
	%\begin{footnotesize}
	\fbox{
		\begin{minipage}{0.97\linewidth}
			\begin{array}[t]{lcl}
				\mathit{Prog} &::= &\Ds\ e\\
				D &::= &\TD\ |\ \CD\\
				\Name &::= &t\ |\ C\\
				\TD &::=  &t = E\\
				\CD &::= &C = E\\
				E &::= &\Body\ |\ t\ |\ E + E\ |\ E[\mathtt{makeAbstract}\ m]\\
				\Body &::= &\{\mathtt{interface}?\ [\Cs]\ \Ms \}\\
%				T &::= &C\\
				M &::= &\MH\ e?;\\
				\MH &::= &S\ \methodKW\ C\ m(C_1\ x_1 \dots C_n\ x_n)\\
				e &::= &x\ |\ e.m(\es)\ |\ \newKW\ C(\es)\\
				\E_v &::= &[].m(es)  \ |\ v.m(vs\ []\ es)\ |\ \newKW\ C(\vs\ []\ \es)\\
				v &::= &\newKW\ C(\vs)\\
				\Gamma &::= &x_1:C_1 \dots x_n:C_n\\
				S &::= & \dots e.g.\ \mathtt{Pre}:P\ \mathtt{Post}:P\\
				P &::= & \dots e.g.\ \text{First order logic}
			\end{array}
	\end{minipage}}
	%\end{footnotesize}
	\caption{Syntax of the trait system}
	\label{f:syntax}
\end{figure}

%To give an example, we want to typecheck that the \Q@maxElement@ method is correctly implemented. In the example, we use an if-then-else construct. This is not part of our object-oriented core language, but used to simplify the example. We could implement conditional statements with lambdas, resulting in similar and also provable code.
%The check starts with the \textsc{MOk} rule that will use the rules for expressions to typecheck the method body and to establish the proof obligation. If we assume the precondition of \Q@maxElement@, 
%the postcondition has to hold after executing the method's body. The proof obligation of the body is assembled by the \textsc{Method} rule. Because this is too complex to show in detail, we roughly explain the proof in the following. We have to satisfy the preconditions of the \Q@accessHead@ and the \Q@maxTail@ method so that they can be called. Then, the postcondition of each branch is assumed, and we have to deduce the postcondition of \Q@maxElement@. In the then branch, we know from the guard that the current element is the biggest in the list, so the postcondition holds. In the else branch, we know that some element is bigger than the current one, so one element in the tail is the searched maximum. This is similar to the postcondition of \Q@maxTail@.  

\subsubsection{Typing Rules}

In our type system, we have a typing context $\Gamma ::= x_1:C_1\dots x_n:C_n$ which assigns types $C_i$ to variables $x_i$.
We define typing rules for our three kinds of expressions: $x$, method calls, and object initialization. We combine typing and verification in our type checking $	\Gamma \vdash e : C \dashv P_0 \models P_1$. This judgment can be read as: under typing context $\Gamma$, the expression $e$ has type $C$, where under the knowledge $P_0$ we need to prove $P_1$. The knowledge $P_0$ is our collected information that we use to prove a method correct.
That means, in our typing rules, we collect the knowledge about the parameters and expressions in a method body to verify that this method body fulfills the specification defined in the method header. The verification obligation $P_1$ should follow from the knowledge $P_0$.

We check if methods are well-typed with judgments of form
$\Ds; \Name \vdash M : \OK$. This judgment can be read as: in the definition table, the method $M$ defined under the definition \Name\ is correct.
The typing rules of Fig.~\ref{f:typing} are explained in the following. 
The first four rules type different expressions and collect the information of these expressions to prove with rule \textsc{MOK} that a method fulfills its specification. In the rule \textsc{MOK} with keyword \textbf{verify}, we call a verifier to prove each method once. Abstract methods (\textsc{AbsOK}) are always correct. Rule \textsc{BodyTyped} ensures that all methods in a body are correctly typed.

\begin{figure}[t]
	\centering
	\begin{scriptsize}
		\begin{mathpar}
			\inferrule
			{
			}{
				\Gamma \vdash x: \Gamma(x) \dashv \result: \Gamma(x)\ \&\ \result = x \models \mathit{true}}
			\quad (\textsc{x})
			
			\inferrule
			{
				S\ \methodKW\ C\ m (C_1\ x_1 \dots C_n\ x_n) \_;\ \in\ \mathit{methods(C_0)}\\
				\Gamma \vdash e_0 : C_0 \dashv P_0 \models P'_0\
				\dots\
				\Gamma \vdash e_n : C_n \dashv P_n \models P'_n\\
				x'_0 \dots x'_n\ \mathit{fresh}\\
				S' = S[\this:=x'_0,\ x_1 := x'_1,\ \dots,\ x_n:= x'_n]\\
				P = (\result:C\ \&\ P_0[\result:=x'_0]\ \&\ \dots\ \&\ P_n[\result:=x'_n] \&\ (\Pre(S') \implies \Post(S')))
			}{
				\Gamma \vdash e_0.m(e_1\dots e_n) : C \dashv P \models P'_0\ \&\ \dots\ \&\ P'_n\ \&\ \Pre(S')}
			(\textsc{Method})
			
			\inferrule
			{
				\Gamma \vdash e_1 : C_1 \dashv P_1 \models P'_1\
				\dots\
				\Gamma \vdash e_n : C_n \dashv P_n \models P'_n\\
				\mathit{getters}(C) = S_1\ \methodKW\ C_1\ x_1();\ \dots\ S_n\ \methodKW\ C_n\ x_n();\\
				x'_1 \dots x'_n\ \mathit{fresh}\\
				S'_i = S_i[\this:=\result]\\
				P''_i = (P_i[\result:=x'_i]\ \&\ (\Pre(S'_i) \implies \result.x_i()=x'_i))\\
				P = (\result:C\ \&\ P''_1\ \&\ \dots\ \&\ P''_n)
			}{
				\Gamma \vdash \newKW\ C(e_1 \dots e_n) : C \dashv P \models P'_1\ \&\ \dots\ \&\ P'_n\ \&\ \Pre(S'_1)\ \&\ \dots \&\ \Pre(S'_n)}
			\quad (\textsc{New})
			
			\inferrule
			{
				\Gamma \vdash e : C' \dashv P \models P'\\
				C'\ instanceof\ C
			}{
				\Gamma \vdash e: C \dashv P \models P'}
			\quad (\textsc{sub})
			
			\inferrule
			{
				\Gamma = \this:\Name,\ x_1:C_1,\ \dots,\ x_n:C_n\\
				\Gamma \vdash e : C \dashv P \models P'\\\\
				\textbf{verify}\ Ds \vdash (\Gamma\ \&\ \Pre(S)\ \&\ P) \models (P'\ \&\ \Post(S))
			}{
				\Ds;\ \Name \vdash S\ \methodKW\ C\ m (C_1\ x_1 \dots C_n\ x_n)\ e;\ : \OK}
			\quad (\textsc{MOK})
			
			\inferrule
			{
			}{
				\Ds;\ \Name \vdash S\ \methodKW\ C\ m(C_1\ x_1 \dots C_n\ x_n);\ : \OK}
			\quad (\textsc{AbsMOK})
			
			\inferrule
			{
				\Body = \{\interface?\ [\Cs]\ M_1 \dots M_n\}\\\\
				\Ds; \Name \vdash M_1 : \OK \dots \Ds; \Name \vdash M_n : \OK
			}{
				\Ds; \Name \vdash \Body\ : \OK}
			\quad (\textsc{BodyTyped})	
		\end{mathpar}
	\end{scriptsize}
	\caption{Expression typing rules of \TraitCbC}
	\label{f:typing}
\end{figure}

\begin{description}
	\item[\textsc{x}] As usual, the type of a variable is stored in the environment $\Gamma$.
	From the verification perspective, we do not need to prove anything to be allowed to use a variable; thus we use $\mathit{true}$.
	We know that the result of evaluating a variable is the value of such variable, and that such value is of the type of the variable; thus we have $\mathit{\result:\Gamma(x)\ \&\ \result = x}$. The $\result$ is the returned value of evaluating this expression, and $\mathit{variable:type}$  is a predicate in our system.
	As you can notice, we are assuming that our parametric logic supports at least a logical \textit{and} ($\&$); but of course other ways to merge knowledge could work too.
	
	\item[\textsc{Method}] As usual, to type a method call, we inductively type the receiver and
	all the parameters. In this way, we obtain all the types $C_0 \dots C_n$, all the knowledge
	$P_0 \dots P_n$, and all the verification obligations $P'_0 \dots P'_n$.
	Inside of all conditions  $P_i \models P'_i$ we call the result of $e_i$ $\result$.
	We cannot simply merge the knowledge of $P_0 \dots P_n$, since their
	$\result$ refers to different concepts. Thus, we chose fresh $x'_0 \dots x'_n$ variables, and we rename $\result$ of $P_i$
	and $P'_i$ into $x'_i$.
	Similarly, $S'$ is the specification of the method adapted using $x'_0 \dots x'_n$.
	
	The verification obligation of course contains
	all the obligations of the receiver and the parameters, but also
	requires the precondition of the method to hold.
	
	The knowledge contains the knowledge of the receiver and the
	parameters, and the method specification in implication form.
	Naively, one could expect that since the precondition is already in
	the obligation we could simply add the postcondition to the
	knowledge.
	This would be unsound. By using the specification in implication form, the system prevents circular reasoning: we could otherwise
	use the postcondition to prove the precondition.
	Instead, when the system shows that the precondition of $S'$ holds,
	it can assume the postcondition of $S'$.
	Similar to logical \textit{and} above, we are assuming that our parametric logic supports at
	least logical \textit{implication}, but of course other forms of logical consequence could work too.
	
	Note that the postcondition will contain information about the
	result of the method body as information on the $\result$ variable.
	
	\item[\textsc{New}] As usual, to type an object instantiation, we inductively type all the
	parameters. In this way we obtain all the types $C_1 \dots C_n$, all the knowledge
	$P_1 \dots P_n$, and all the verification obligations $P'_1 \dots P'_n$.
	As we did for \textsc{Method} we use fresh variables to be able to compose predicates.
	
	As we mentioned above, we rely on abstract state operations to represent state:
	that is, all the abstract methods in $C$ need to be of form $S_i\ \mathtt{method}\ C_i\ x_i();$
	where $\this.x_i()$ returns the value of field $x_i$, that in turn was
	initialized with the result of expression $e_i$. The function $\mathit{getters}(C)$ returns all methods of this form.
	
	Knowledge $P''_i$ contains the knowledge of $P_i$ (from expression $e_i$) and it links
	such knowledge to the result of calling method $\result.x_i()$, so that
	calling a getter on the created object will return the expected value.
	However, the information is conditional over verifying the precondition
	of such getter.
	Note that we do not need to add the knowledge of the postcondition of
	$x_i()$ here; this will be handled by the \textsc{Method} rule when $x_i()$ is called.
	
	Knowledge $P$ is simply merging the accumulated knowledge; while the final obligation
	in addition to merging the accumulated obligations also requires that the
	precondition of all the getters hold. In this way the getter
	preconditions behave like the precondition of the constructor.
	By requiring those preconditions, we ensure that we can call the
	getters on all the created objects.
	\item[\textsc{Sub}] The subsumption rule is standard. We allow subtyping between class names. Note that we do not apply weakening and strengthening of conditions here.
\end{description}

Besides of typing correct programs, the typing rules of Trait-CbC have the goal to verify the correctness of method implementations.
The following rules check whether a method or a $\mathit{Body}$ are correct. The check for a correct method declaration in \textsc{MOK} calls a program verifier to verify the correctness. We need just one verifier call for the verification of each method because the rules above collected all needed knowledge and obligations.

\begin{description}
	\item[\textsc{MOK}] 
	In \textsc{MOK}, we construct a $\Gamma$, and we type the method body, obtaining
	knowledge $P$ and obligation $P'$.
	The program verifier will know the type information of $\Gamma$, the premise of the method,
	and the knowledge $P$, and will prove the obligation $P'$ and the postcondition of the method. This verification in the typing rule is indicated by the keyword \textbf{verify}. 
	Here, we use implication, but a different program verifier may use a
	different form of logical consequence.
	The program verifier can access the specification of all the other methods since
	we also provide the declaration table.
	\item[\textsc{AbsMOK}] Abstract methods are correctly typed.
	\item[\textsc{BodyTyped}] A $\Body$ is correctly typed, if all the methods in the declaration of the $\Body$ are correctly typed.
\end{description}

%\paragraph{Example on Typing}
%To illustrate an example with explicit proof goals, we use the verification of the method \Q@accessHead@. In the body of \Q@accessHead@, the method \Q@list.element()@ is called. It is a getter-method with the simple postcondition \Q@result == list.element()@.
%The proof obligation created by the rule \textsc{MOK} is the following. In the typing rule, the expression $e$ is the call of the method \Q@list.element()@. We use the rule \textsc{Method} to type this expression, and we obtain the following proof obligation:
%
%$\Gamma\ \&\ $
%\Q@result@
%$:$
%\Q@Num@
%$\ \&$
%\Q@list@
%$ = x'_0\ \&\ (\mathit{true}$
%$ \Rightarrow $
%\Q@result@
%$ = x'_0$
%\Q@.element()@
%$)\ \models  $
%\Q@result@
%$ = $
%\Q@list.element()@
%%result = x_0.element()  & x_0 = list
%
%We prove that the body of \Q@accessHead@ satisfies the pre-/postcondition specification. The precondition of the called method \Q@list.element()@ is satisfied, as it is $\mathit{true}$. Therefore, we can use the postcondition of this method to show that the postcondition of \Q@accessHead@ holds. If we replace $x'_0$ with \Q@list@, we have the same condition in our pre- and postcondition and can close the goal. Thus, the method \Q@accessHead@  is proven correct.

\subsubsection{Reduction Rules}
\label{sec:reduction}

We formulate three reduction rules for our system to evaluate input expressions to final values. We introduce an evaluation context $\E_v$ in our syntax in Fig.~\ref{f:syntax} to define the order of evaluation. The rules of Fig.~\ref{f:reduction} are explained in the following.

\begin{figure}[t]
	\centering
	\begin{scriptsize}
		\begin{mathpar}
			\inferrule
			{
				\Ds \vdash e \rightarrow e'
			}{
				\Ds \vdash \E_v[e] \rightarrow \E_v[e']}
			\quad (\textsc{$Ctx$})
			
			\inferrule
			{
				S\ \methodKW\ C\ m(C_1\ x_1,\ \dots,\ C_n\ x_n)\ e;\ \in\ \mathit{methods(C)}
			}{
				\Ds \vdash \newKW\ C(\vs).m(v_1\dots v_n) \rightarrow e[\this=\newKW\ C(\vs),\ x_1=v_1,\ \dots,\ x_n=v_n]}
			\quad
			(\textsc{mcall})
			
			\inferrule
			{
				\mathit{abs}(\Ds(C))=S_1\ \methodKW\ C_1\ x_1();\ \dots\ S_n\ \methodKW\ C_n\ x_n();
			}{
				\Ds \vdash \newKW\ C(v_1 \dots v_n).x_i() \rightarrow  v_i}
			\quad (\textsc{getter})
		\end{mathpar}
	\end{scriptsize}
	\caption{Reduction rules of \TraitCbC}
	\label{f:reduction}
\end{figure}

\begin{description}
	\item[\textsc{Ctx}] This is the conventional contextual rule, allowing the execution of subexpressions.
	\item[\textsc{Mcall}] We reduce a method call to an expression $e$, where the receiver is replaced with \newKW\ $C(\vs)$, and each parameter $x_i$ with the actual value $v_i$.
	We also ensure that the method is declared in the class $C$.
	\item[\textsc{Getter}] In our formalism, abstract methods without arguments represents getters. Notation $\mathit{abs(Body)}$ returns the set of all abstract methods in $\Body$. A valid class can only have
	abstract methods without arguments, and they will all represent getters.
\end{description}

\subsubsection{Flattening Semantics}
When we implement methods in several traits, we have to check that these traits are compatible when they are composed. This process to derive a complete class from a set of traits is called flattening. We follow the traditional flattening semantics~\cite{ducasse2006traits}.
A class that is defined by composing several traits is obtained by flattening rules. All methods are direct members of the class~\cite{ducasse2006traits}. Overall, our flattening process works as a big step reduction arrow, where we reduce a trait expression into a well-typed and verified body.

%These classes will be used for type checking and verification purposes.

%using names of them methods

%${}_{}\quad \mathit{allMeth}(m,\ Bodies) = \{MH;\ |\ Body\ in\ Bodies,\ Body(m)=MH;\}$

To introduce our flattening rules in Fig~\ref{f:flattening}, we first define the helper functions.
The function $\mathit{allMeth}$ collects all method headers with the same name as $m$ in all input bodies (Definition~\ref{def:allmeth}).
When two $\mathit{Body}$s are composed (Definition~\ref{def:bodyplus}), the implemented interfaces are united and the methods are composed. 
The composition of methods (Definition~\ref{def:msplus}) collects methods that are only defined in one of the input sets. If a method is in both sets, it is composed (Definition~\ref{def:mplus}). Here, we distinguish four cases. If one method is abstract and the other is concrete, we have to show that the precondition of the abstract method implies the precondition of the concrete method. Additionally, the postcondition of the concrete one has to imply the postcondition of the abstract one. This is similar to Liskov's substitution principle~\cite{liskov}.
The second case is the symmetric variant of the first case. In the third and fourth case, two abstract methods are composed. Here, the specification of one abstract method has to imply the specification of the other abstract method such that an implementation can still satisfy all specifications of abstract methods. If both methods are concrete, the composition is correctly left undefined. This composition error can be resolved by making one method $m$ abstract in the $\Body$, as defined in Definition~\ref{def:bodyabs}. The resulting $\Body$ is similar with the difference that the implementation of the method $m$ is omitted.
The flattening rules in Fig.~\ref{f:flattening} are explained in the following in detail. In these rules, a set of traits is flattened to a declaration containing all methods. If abstract and concrete methods with the same name are composed, Definitions \ref{def:bodyplus}-\ref{def:mplus} are used to guarantee correctness of the composition.

\begin{definition}[All Methods]\label{def:allmeth}
	$\mathit{allMeth}(m,\ \mathit{Bodys}) =$\\
	${}_{}\quad \{\MH;\ |\ \Body\ \in\ \mathit{Bodys},\ \Body(m)=\MH;\}$
\end{definition}
	\begin{definition}[Body Composition]\label{def:bodyplus} $\Body_1 + \Body_2 = \Body$\\
		${}_{} \{ \interface?\ [\Cs_1]\ \Ms_1 \} + \{ \interface?\ [\Cs_1]\ \Ms_1 \} =$\\ 
		${}_{}\{ \interface?\ [\Cs_1\ \cup \Cs_2]\ \Ms_1+\Ms_2 \}$
	\end{definition}
	\begin{definition}[Methods Composition]\label{def:msplus} $\Ms_1+\Ms_2=\Ms$\\
		${}_{}\bullet (M\ \Ms_1) + \Ms_2 = M\ (\Ms_1+\Ms_2)$\\
		${}_{}\quad  \mathit{if\ methName(M)}\notin\mathit{dom (\Ms_2)}$\\
		${}_{}\bullet (M_1\ \Ms_1) + (M_2\ \Ms_2) = M_1+M_2\ (\Ms_1+\Ms_2)$\\
		${}_{}\quad \mathit{if\ methName(M_1)= methName(M_2)}$\\
		${}_{}\bullet \emptyset + \Ms = \Ms$
	\end{definition}
	\begin{definition}[Method Composition]\label{def:mplus} $M_1+M_2=M $\\
		${}_{}\bullet S\ \methodKW\ C\ m(C_1\ x_1 \dots C_n\ x_n)\ e;\ +\ S'\ \methodKW\ C\ m(C_1\ \_\dots C_n\ \_);$\\
		${}_{}\quad =\ S\ \methodKW\ C\ m(C_1\ x_1 \dots C_n\ x_n)\ e;$\\
		${}_{}\quad \mathit{if}\  \Pre(S')\ \mathit{implies}\ \Pre(S)\ and\ \Post(S)\ \mathit{implies}\ \Post(S')$\\
		${}_{}\bullet \MH_1;\ +\ \MH_2\ e;\quad  =\quad  \MH_2\ e;\ +\ \MH_1;$\\
		${}_{}\bullet S\ \methodKW\ C\ m(C_1\ x_1\dots C_n\ x_n);\ +\ S'\ \methodKW\ C\ m(C_1\ \_\dots C_n\ \_);$\\
 		${}_{}\quad =\ S\ \methodKW\ C\ m(C_1\ x_1\dots C_n\ x_n);$\\
		${}_{}\quad \mathit{if}\  \Pre(S')\ \mathit{implies}\ \Pre(S)\ \mathit{and\ Post(S)\ implies\ Post(S')}$\\
		${}_{}\bullet S\ \methodKW\ C\ m(C_1\ x_1\dots C_n\ x_n);\ +\ S'\ \methodKW\ C\ m(C_1\ \_\dots C_n\ \_);$\\
		${}_{}\quad =\ S'\ \methodKW\ C\ m(C_1\ x_1\dots C_n\ x_n);$\\
		${}_{}\quad \mathit{if\  (Pre(S)\ implies\ Pre(S')\ and\ Post(S')\ implies\ Post(S))}$\\
		${}_{}\quad \mathit{and\ not\  (Pre(S')\ implies\ Pre(S)\ and\ Post(S)\ implies \Post(S'))}$
		%${}_{}\quad\bullet S\ method\ C\ m(C_1\ x_1\dots C_n\ x_n); + S'\ method\ C\ m(C_1\ \_1\dots C_n\ \_n);$\\
		%${}_{}\quad\quad = S''\ method\ C\ m(C_1\ x_1\dots C_n\ x_n);$\\
		%${}_{}\quad\quad\quad Pre(S'') = Pre(S) \cup Pre(S')$\\
		%${}_{}\quad\quad\quad Post(S'') = Post(S) \cup Post(S')$
	\end{definition}
	\begin{definition}[Body Abstraction]\label{def:bodyabs} $\Body[\mathtt{makeAbstract}\ m]$\\
	${}_{}\quad \{ [\Cs]\ \Ms_1\ S\ \mathtt{method}\ C\ m(\mathit{Cxs}) \_;\  \Ms_2\}[\mathtt{makeAbstract}\ m] $\\ 
	${}_{}\quad = \{ [\Cs]\ \Ms_1\ S\ \mathtt{method}\ C\ m(\mathit{Cxs});\  \Ms_2\}$
	\end{definition}

\begin{figure}[!t]
	\centering
	\begin{scriptsize}
	\begin{mathpar}
		\inferrule
		{
			D'_1 \dots D'_n \vdash D_1 \Downarrow D'_1\\ \dots\\ D'_1 \dots D'_n \vdash D_n \Downarrow D'_n
		}{
			D_1 \dots D_n \Downarrow D'_1 \dots D'_n}
		\quad (\textsc{FlatTop})
		
		\inferrule
		{
			\Ds;\ \Name \vdash E \Downarrow \Body\\
			\mathtt{if}\ \Name\ \mathtt{of\ form}\ C\ \mathtt{then}\ \mathit{abs(Body)} = S\ T\ x_1(); \dots S\ T\ x_n();
		}{
			\Ds \vdash \Name = E \Downarrow \Name = \Body}
		\quad (\textsc{DFlat})
		
		\inferrule
		{
			\Body = \{\interface?\ [\Cs]\ M_1 \dots M_n\}\\\\
			\Body' = \{\interface?\ [\Cs]\ M_1 \dots M_n\ \Ms\}\\
			\Ms = \{\Sigma \mathit{allMeth}(\Ds,\ \Cs,\ m)\ |\  m \in \mathit{dom(Cs)}\ \mathit{and}\ m \notin \mathit{dom(Body)} \}\\
			\Ds;\ \Name \vdash \Body' : \OK
		}{
			\Ds;\ \Name \vdash \Body \Downarrow \Body'}
		 \quad (\textsc{BFlat})
		
		\inferrule
		{
		}{
			\Ds;\ \Name \vdash t \Downarrow \Ds(t)}
		\quad (\textsc{tFlat})
		
		\inferrule
		{
			\Ds;\ \Name \vdash E_1 \Downarrow \Body_1\\
			\Ds;\ \Name \vdash E_2 \Downarrow \Body_2\\
		}{
			\Ds;\ \Name \vdash E_1 + E_2 \Downarrow \Body_1 + \Body_2}
		\quad (\textsc{+Flat})
		
		\inferrule
		{
			\Ds;\ \Name \vdash E \Downarrow \Body\\
			\Body = \{[\Cs]\ \overline{M}_1\ S\ \mathtt{method}\ C\ m(C_1\ x_1 \dots C_n\ x_n) \_;\ \overline{M}_2\}\\
			\Body' = \{[\Cs]\ \overline{M}_1\ S\ \mathtt{method}\ C\ m(C_1\ x_1 \dots C_n\ x_n);\ \overline{M}_2\}
		}{
			\Ds;\ \Name \vdash E[\mathtt{makeAbstract}\ m] \Downarrow \Body'}
		\quad (\textsc{AbsFlat})
	\end{mathpar}
	\end{scriptsize}
	\caption{Flattening rules of \TraitCbC}
	\label{f:flattening}
\end{figure}

\begin{description}
	\item[\textsc{FlatTop}] The first rule flattens a set of declarations $D_1\dots D_n$ to a set $D'_1 \dots D'_n$.	We express this rule in a non-computational way:
	we assume to know the resulting $D'_1 \dots D'_n$, and we use them as a guide to
	compute them.
	Note that if there is a resulting $D'_1 \dots D'_n$ then it is unique;
	flattening is a deterministic process and $D'_1 \dots D'_n$ are used only
	to type check the results. They are not used to compute the shape of
	the flattened code.
	
	Non computational rules like this are common with nominal type
	systems~\cite{igarashi2001featherweight} where the type signatures of all classes and methods can
	be extracted before the method bodies are verified.
	
	\item[\textsc{DFlat}] This rule flattens an individual definition by flattening the trait expression.
	When the flattening produces a class definition, we also check that the body
	denotes an instantiable class; a class whose only abstract methods are
	valid getters. The function $\mathit{abs(Body)}$ returns the abstract methods.
	
	\item[\textsc{BFlat}] It may look surprising that the $\Body$ does not flatten to itself. This represents what happens in most programming languages, where
	implementing an interface implicitly imports the abstract signature
	for all the methods of that interface.
	In the context of verification also the specification of such interface methods is imported.
	In concrete, $\Body'$ is like $\Body$, but we add $\Ms$ by collecting all the methods
	of the interfaces that are not already present in the $\Body$.
	
	Moreover, we check that all the methods defined in the class respect
	the typing and the specification defined in the interfaces:
	if a class has $S\ $\Q@method Foo@ \Q@foo();@ or 
	$S\ $\Q@method Foo foo() e;@
	and there is a $S'\ $\Q@method Foo foo();@ in the interface,
	then $S$ must respect the specification $S'$.
	The system then checks that the $\Body$ is well-typed and verified by calling
	$\Ds;\ \Name \vdash M_i : \OK$
	
	%	$forall\ C\ \in Cs,\ S'\ method\ C_0\ m(C_1\ x_1 \dots C_n\ x_n); \in Ds(C)$
	%	$where\ Body(m)=MH\ then\ MH=S\ method\ C_0\ m(C_1\ x_1\dots C_n\ x_n)\_$
	%	$Pre(S')\ implies\ Pre(S)\ and\ Post(S)\ implies\ Post(S')$
	
	\item[\textsc{TFlat}] A trait $t$ is flattened to its declaration $\Ds(t)$.
	
	\item[\textsc{+Flat}] The composition of two expression $E_1$ and $E_2$, where both expressions are first reduced to $\Body_1$ and $\Body_2$, results in the composition of these bodies as defined in Definition~\ref{def:bodyplus}.
	
	\item[\textsc{AbsFlat}] An expression $E$ where one method $m$ is made abstract flattens to a $\Body'$. We know that $E$ flattens to $\Body$. The only difference between $\Body$ and $\Body'$ is that the one method $m$ is abstract in $\Body'$. In $\Body$, the method can be abstract or concrete.
\end{description}

\subsubsection{Soundness of \TraitCbC}

In this subsection, we formulate the main result of the \TraitCbC approach. We prove soundness of the flattening process with a parametric logic.
We claim that if you have a language without code reuse and with sound and modular post-hoc verification
then the language supports CbC simply by adding traits to the language. That is, traits intrinsically enable a CbC program construction approach.

%For this claim, we need two assumptions. %Note how these assumptions are the 'if' of our claim.
%1)
% The parametric logic is sound. 
%2) We have a program with just classes. If a method is well-typed and verified, then
%the execution of the method respects its contract.

To prove soundness of the construction approach of \TraitCbC (Theorem~\ref{theorem:cbcsound}: Sound CbC Process) as exemplified in Section~\ref{sec:example}, we have to show that the flattening process is correct (Theorem~\ref{theorem:soundness}: General Soundness). 
%Any correct program that is flattened should be correct after the flattening.
%Any program that can be successufully flattened is well-typed and verified.
In turn, to prove General Soundness, we need two lemmas which state that the composition of traits is correct (Lemma~\ref{lemma:flattening}) and that a trait after the $\mathtt{makeAbstract}$ operation is still correct (Lemma~\ref{lemma:makeAbstract}).

In Lemma~\ref{lemma:flattening}, we have well-typed definitions $\Ds$, and two well-typed and verified traits in $\Ds$, and the resulting trait/class is also well-typed and verified.

\begin{lemma}[Composition correct]\label{lemma:flattening}\quad\\
If $\Ds(t1) = \Body_1$,
$\Ds(t2) = \Body_2$, 
$\Ds(\Name) = \Body$,
$\Ds; t_1 \vdash \Body_1 : \OK$,
${}_{}\quad$ $\Ds; t_2 \vdash \Body_2 : \OK$,
and $\Body_1 + \Body_2=\Body$,\\*
then $\Ds; \Name \vdash \Body : \OK$
\end{lemma}

\begin{proof}
	We prove by contradiction. 
	We assume the resulting $\Body$ is ill typed.
	By definition of \textsc{BodyTyped}, it means that one of the methods cannot be
	typed with either \textsc{AbsMOK} or \textsc{MOK}.
	The list of methods that need to be typed is obtained by Definition~\ref{def:bodyplus}.
	
	Abstract methods can only be typed with \textsc{AbsMOK} and are never wrong.
	Implemented methods can only be typed with \textsc{MOK}.
	If $\Gamma \vdash e : C \dashv P \models P'$
	or the other precondition $\textbf{verify}\ \Ds \vdash (\Gamma\ \&\ \Pre(S)\ \&\ P) \models (P'\ \&\ \Post(S))$ does not hold, 
	it means that there was a method $m_i$ with expression $e_i$ in $\Body_1$ (or
	symmetrically for $\Body_2$)
	that was well-typed under $\Ds; t_1\vdash \Body_1$. That means that all of its implemented methods were well-typed and verified.
	Typing $e_i$ produces
	$P_i \models P'_i$ by using a $\Gamma_{t1}$ containing $\this:t_1$.
	
	If $\Ds; \Name \vdash \Body : \OK$ is not applicable, the same
	expression $e_i$
	was typed using a $\Gamma_{\Name}$ containing $\this:\Name$. It produced $P''_i \models P'''_i$
	so that $\textbf{verify}\ \Ds \vdash (\Gamma_{\Name}\ \&\ \Pre(S)\ \&\ $ $P''_i) \implies (P'''_i\ \&\ \Post(S))$ does not hold.
	We know that $\textbf{verify}\ \Ds \vdash (\Gamma_{t1}\ $ $\&\ \Pre(S)\ \&\ P_i) \implies (P'_i\ \&\ \Post(S))$ holds by our assumption.
	By Definition~\ref{def:mplus}, the contracts of the methods in $\Body$ are simply stronger than the contracts of the methods in $\Body_1$.
	The only difference between $P''_i \models P'''_i$ and $P_i \models P'_i$ is in
	the contracts of methods called on $\this$.
	Assuming that our parametric logic implication is transitive,
	we know that $\textbf{verify}\ \Ds \vdash (\Gamma_{t1}\ \&\ \Pre(S)\ \&\ P_i) \implies (P'_i\ \&\ \Post(S))$ entails $\textbf{verify}\ \Ds \vdash (\Gamma_{\Name}\ \&\ \Pre(S)\ $ $\&\ P''_i) \implies (P'''_i\ \&\ \Post(S))$,
	thus we reach a contradiction.
\end{proof}

Lemma~\ref{lemma:makeAbstract} shows that if we have a well-typed and verified trait, the
operation $\mathtt{make-}$ $\mathtt{Abstract}$ results in a trait/class that is also well-typed and verified.

\begin{lemma}[MakeAbstract correct]\label{lemma:makeAbstract}\quad\\
If $\Ds(t) = \Body$,
$\Ds(\Name) = \Body'$,
$\Ds; t \vdash \Body : \OK$,\\* 
${}_{}\quad$ and $\Body[\mathtt{makeAbstract}\ m] = \Body'$,\\*
then $\Ds; \Name$ $\vdash \Body' : \OK$
\end{lemma}

\begin{proof}
	We prove by contradiction. 
	We assume the resulting $\Body'$ is ill typed.
	By definition of \textsc{BodyTyped}, it means that one of the methods cannot be
	typed with either \textsc{AbsMOK} or \textsc{MOK}.
	The list of methods that need to be typed is obtained by Definition~\ref{def:bodyplus}.
	
	Abstract methods can only be typed with \textsc{AbsMOK} and are never wrong.
	We know that $\Body$ is typable by our assumption. The only difference between $\Body$ and $\Body'$ is that the method $m$ is made abstract. As we have seen for Lemma~\ref{lemma:flattening}, we are typing $\Body'$ in a different $\Gamma$. This case is even simpler than Lemma~\ref{lemma:flattening} because $\Body$ and $\Body'$ have exactly the same specifications. The abstract method $m$ and thus $\Body'$ cannot be ill typed.
\end{proof}

With these Lemmas, we can prove Theorem~\ref{theorem:soundness}. Given a sound and modular verification language, then all programs that flatten are well-typed and verified.
In a modular verification language, a method can be fully verified using only the information contained in the method declaration and the specification of any used method.
Moreover, our parametric logic must support at least a commutative and associative \textit{and} (but of course other ways to merge knowledge could work too) and a transitive \textit{implication} (but of course other forms of logical consequence could work too).

\begin{theorem}[General Soundness]\label{theorem:soundness}\quad\\
For all programs $\Ds$ where $\Ds$ flattens to $\Ds'$, and $\Ds'$ is well-typed;\\*
that is, fo rall $\Name=\Body \in \Ds'$,
we have $\Ds';\ \Name \vdash \Body : \OK$.
\end{theorem}

\begin{proof}
	By induction on the size of $\Ds$, and by induction on cases of $E$ (the
	applied flattening rule for $E$).
	\begin{itemize}
		\item $\Body$ only flattens if the $\Body$ can be shown to be well-typed.
		\item $t$ only reads a trait from the already verified $\Ds'$.
		\item $\Body_1+\Body_2$ is correct with Lemma~\ref{lemma:flattening}. The lemma can be applied directly, if $E$ is of depth one (e.g., $\Body_1 + \Body_2$). If $E$ is more complex, we have to apply other cases of this case analysis.
		\item $\mathtt{makeAbstract}$ is handled similarly using Lemma~\ref{lemma:makeAbstract}.
		\item By the flattening relation, we know that $\Body_1$ and $\Body_2$ are well-typed in $\Ds$.
		If we start from a program containing only well-typed and verified traits,
		any new class we can define by just composing those traits is well
		typed and verified. \qedhere
	\end{itemize}
\end{proof}

We now show that the \TraitCbC approach is sound. Theorem~\ref{theorem:cbcsound} states that starting with one abstract method and a set of verified traits, the composed program is also verified.

\begin{theorem}[Sound CbC Process]\label{theorem:cbcsound}\quad\\
	Starting from a fully abstract specification $t_0$, and some construction steps $t_1 \dots t_n$, we can write $C = t_0 + \dots + t_n$ as our whole CbC approach, where $t_0 + t_1$ is the application of the first construction step. 
	If we use CbC to construct programs, we can start from verified atomic units and get a verified result. Formally, if $t_0 = \{\MH\}\
	t_1 = \{\Ms_1\}\
	\dots\
	t_n = \{\Ms_n\}$ are well-typed, and\\
	$\begin{array}{lcl}
		t_0 = \{\MH\}&&t_0 = \{\MH\}\\
		t_1 = \{\Ms_1\}\
		\dots\
		t_n = \{\Ms_n\} &\quad\Downarrow \quad\quad & t_1 = \{\Ms_1\}\
		\dots\
		t_n = \{\Ms_n\}\\
		C = t_0 + \dots + t_n && C = \Body\\
	\end{array}$\\
	then $C=\Body$ is well-typed.
\end{theorem}
\begin{proof}
	This is a special case of Theorem~\ref{theorem:soundness}.%, therefore this theorem is valid.
\end{proof}

Theorem~\ref{theorem:cbcsound} shows clearly that trait composition intrinsically enables a CbC approach: 
An object-oriented programming language with traits
and a corresponding specification language supports an incremental CbC approach.

\subsection{Proof-of-Concept Implementation}
\label{sec:impl}

In this section, we describe the implementation, which instantiates \TraitCbC in Java with JML~\cite{leavens1998jml} as specification language and KeY~\cite{ahrendt2016deductive} as verifier for Java code. Our trait implementation is based on interfaces with default implementation. Our open source tool is implemented in Java and integrated as plug-in in the Eclipse IDE.\footnote{Tool and evaluation at \url{https://doi.org/10.5281/zenodo.7766635}}
Besides this prototype, other languages with a suitable verifier, such as Dafny~\cite{leino2010dafny} and OpenJML~\cite{cok2011openjml}, can also be used to implement \TraitCbC.

%As in the example of Section~\ref{sec:example}, each method in a trait is specified in JML with a pre- and postcondition. 
\hl{In Listing~\ref{code:exampleJava}, we show the concrete syntax of our implementation. Each method in a trait is specified with JML with the keywords \Q@requires@ and \Q@ensures@ for the pre- and postcondition. 
%The methods are written in Java syntax. 
To verify the correctness of programs, we need two steps. First, we verify the correctness of a method implemented in a trait w.r.t. its specification. 
Second, for trait composition, our implementation checks the correct composition for all methods (cf. Definition~\ref{def:bodyplus}). 
The syntax of trait composition is shown in Listing~\ref{code:exampletraitcompo}. In a tc-file (a file to specify the traits to be composed), the name of the resulting trait is given and the composed traits are connected with a plus operator. In Listing~\ref{code:exampletraitcompo}, trait \Q@MaxElement1@ is composed with trait \Q@MaxElement2@. The trait \Q@MaxElement2@ must implement the methods \Q@accessHead@ and \Q@maxTail@, so that we obtain a correct result in which all methods are implemented.
To verify correctness of the trait composition, it is checked that the specification of a concrete method satisfies the specification of the abstract one with the same signature (cf. Definition~\ref{def:mplus}). These verification goals are sent to KeY, which starts an automatic verification attempt.}

\begin{figure}[ht]
\begin{lstlisting}[caption={Example of a trait in our implementation},captionpos=b,label={code:exampleJava}]
public interface MaxElement1 {
/*@ requires list.size() > 0;
  @ ensures (\forall int n; list.contains(n);
  @ \result >= n) & list.contains(\result);
  @*/
  default public int maxElement(List list) {
    if (list.size() == 1) return accessHead(list);
    if (list.element() >= maxElement(list.tail())) 
      { return accessHead(list); }
    else { return maxTail(list); } }

/*@ requires list.size() > 0;
  @ ensures \result == list.element();
  @*/
  public int accessHead(List list);

/*@ requires list.size() > 1;
  @ ensures (\forall int n; list.tail().contains(n);
  @ \result >= n) & list.tail().contains(\result);
  @*/
  public int maxTail(List list);
}
\end{lstlisting}
\end{figure}

\begin{figure}[ht]
	\begin{lstlisting}[caption={Example of a trait composition},captionpos=b,label={code:exampletraitcompo}]
	ComposedMax = MaxElement1 + MaxElement2
	\end{lstlisting}
\end{figure}

\paragraph{Evaluation}

We evaluate our implementation by a feasibility study. First, we reimplemented an already verified case study in our trait-based language. We used the\linebreak IntList~\cite{scholz2011intlist} case study, which is a small software product line (SPL) with a common code base and several features extending this code base. Here, we can show that our trait-based language also facilitates reuse.
The IntList case study implements functionality to insert integers to a list in the base version. Extensions are the sorting of the list and different insert options (e.g., front /back). We implement five methods that exists in different variants with our trait-based CbC approach.
We implement the case study in different granularities. The coarse-grained version is similar to the SPL implementation we started with~\cite{scholz2011intlist}, confirming that traits are also amenable to implement SPLs as shown by Bettini et al.~\cite{bettini2010implementing}. The fine-grained version implements the five methods incrementally with 12 construction steps. We can reuse 6 of these steps during the construction of method variants.

We also implement three more case studies (BankAccount~\cite{thuem2012familybased}, Email~\cite{hall2005fundamental}, and Elevator~\cite{plath2015elevator}) with \TraitCbC and classic CbC to show that it is feasible to implement object-oriented programs with both approaches. We used \corc~\cite{runge2019tool} as an instance of a classic CbC tool.
We were able to implement 9 classes and verify 34 methods with a size of 1--20 lines of code. For future work, a user study is necessary to evaluate the usability of \TraitCbC in comparison to classic CbC to empirically confirm our stated advantages.

\section{The Different CbC-based Program Construction Approaches in Comparison}
\label{sec:discussion}
In this section, we discuss classic CbC in comparison to \CbCBlock and \TraitCbC.
In Table~\ref{tab:compare}, we summarize how the three approaches compare regarding main aspects of developing correct programs using tool support.
The aspects comprise the programming language, the tool support, the procedure to develop programs, and the verification of the program.

\begin{table}[tb]
	\restoretx
	\centering
	\begin{scriptsize}
	\begin{tabularx}{\linewidth}{nXXX}
		\toprule
		 & Classic CbC	& \CbCBlock & \TraitCbC\\ 
		\hline
		Language & Additional refinement rules for a programming language. Needs specification language.	& Additional refinement rules for a programming language. Introduces a specified block of statements. Needs specification language. & Programming language with traits. Needs specification language.\\
		\hline 
		Tool \newline support & Pen and paper. Some specialized tools available.& Pen and paper. Some specialized tools available. Block instantiation rule relies on post-hoc verification  tools. & Relies on post-hoc verification tools.\\ 
		\hline
		Construc- tion Rules & Specific refinement rules. & Specific refinement rules. & Construction by composition of traits.\\ 
		\hline 
		Correctness/
		Debugging & Guarantees the correctness of each refinement step. & Guarantees the correctness of each refinement step. Refinements can be condensed with the block rules. & Guarantees the correctness of each construction step. Each method is specified so that each constructed method can directly be verified.\\
		\hline
		Proof\newline complexity & Many, but small proofs. & Any granularity of proofs. & Any granularity of proofs.\\
		\hline
		Reuse & 
Refinement steps cannot be reused; only fully implemented methods can.
 & Refinement steps cannot be reused; only fully implemented methods can. & Each verified method in a trait can be reused.\\
 		\hline
 Applications & Focuses on small, but correctness-critical algorithms.
 & Focuses on correctness-critical algorithms. & As \TraitCbC is based on post-hoc verification, it can be used in similar areas where post-hoc verification is used. Traits are beneficial for incremental development and software product lines.\\
		\bottomrule
	\end{tabularx}
	\end{scriptsize}
	\caption{Comparison of \TraitCbC with \CbCBlock and classic CbC}
	\label{tab:compare}
\end{table}

%With trait composition, we can simulate CbC refinement rules when we implement methods similar to the application of CbC refinement rules. 

\noindent\textbf{Language.} All approaches need an underlying programming and specification language.
The defined refinement rules of the classic CbC approach are external to a programming language. That means, each refinement rule introduces some statement of the programming language by transforming the program.
With \CbCBlock and the block-instantiation rule more than one statement of the language can be introduced at once.
\TraitCbC is usable with languages that have traits. Methods can be implemented as defined by the language. No refinement rules are necessary.

\noindent\textbf{Tool Support.}
Tool support is helpful for any of the approaches. For classic CbC, mostly pen and paper is used. There are a few specialized tools such as \corc~\cite{runge2019tool}, ArcAngel~\cite{oliveira2003arcangel}, and SOCOS~\cite{back2009invariant,back2007testing}. These tools force a certain programming procedure on the user because refinement rules must be applied to implement programs.
\CbCBlock is implemented in \corc and extends the set of refinement rules with the new rules for blocks. To verify the correctness of block instantiations, program verifiers can be reused.
There are program verifiers for many languages, such as Java~\cite{ahrendt2016deductive}, C~\cite{cohen2009vcc}, and C\#~\cite{Barnett:2011:SVS:1953122.1953145,barnett2004spec}.
Other languages are integrated with their verifier from the start, e.g., Spec\#~\cite{barnett2004spec} and Dafny~\cite{leino2010dafny}.
For \TraitCbC, we also need a program verifier to prove the correctness of method implementations, but we do not need specialized tools to construct methods, such as \corc.

\noindent\textbf{Construction Rules.}
To construct a program, classic CbC has a strict concept of refinement rules that must be applied to construct a program.
\CbCBlock relaxes this strict guideline to construct programs. 
Programs can be constructed stepwise as with classic CbC, but if desired, any number of refinement steps can be condensed with the block rules. In the extreme case, a whole program can be developed in one step.
\TraitCbC offers this flexibility to construct programs without the need of external refinement rules.
Methods can be developed freely and only need to be composed with respect to their specification. Nevertheless, \TraitCbC supports to construct code in fine-grained steps, which are more amenable for verification than more complex methods.

\noindent\textbf{Correctness/Debugging.}
Classical CbC gives explicit information about the program states before and after execution of each statement by the Hoare triple notation.
The correctness of each applied refinement step is guaranteed by proving the side conditions of the refinement rule. 
Some side conditions are not directly provable because abstract statements in Hoare triples must be concretized first. In the worst case, a problem in the program is found only after some refinement steps. 
The abstract statements in classic CbC are not explicitly specified by the developer. Additional specifications in classic CbC are introduced with some rules such as an intermediate condition in the composition rule. Then, these specifications are propagated through the program to be constructed. Again, due to a possible delayed check of a side condition, a wrong specification is found only after some refinement steps.

If errors occur in the program development process, \TraitCbC gives early and detailed information on the level of verified methods. By specifying the method under development and any abstract method that is called by this method, we can directly verify the correctness of the method under development.
We assume that the introduced abstract methods will be correctly implemented in further refinement steps. With each step, the developer gets closer to the solution until finally all abstract methods are implemented. 
\CbCBlock combines the characteristics of the other two approaches. The refinement rules of classic CbC can be applied, or blocks of statements can be introduced. The specified block is verified similar to a method in \TraitCbC. 

\noindent\textbf{Proof Complexity.}
Classical CbC requires many small proofs to guarantee the correctness of a program. \CbCBlock can condense the proofs into larger proofs using the block refinement rules.
\TraitCbC can have the same granularity and also the same proof effort as classic CbC, since each method implementation can correspond to just one refinement step. 
The advantage of \TraitCbC and \CbCBlock is that developers can freely implement a method body or a block. They must not stick to the same granularity as in the classic CbC refinement rules. Proof complexity can be balanced against verifier calls.

\noindent\textbf{Reuse.}
A fully refined method can be reused in all approaches. 
For \TraitCbC, we can easily reuse even very small units of code, since they are represented as methods in the traits.
In classic CbC and \CbCBlock, no refinement step can be reused. 

\noindent\textbf{Applications.}
\hl{The classic CbC approach does not scale well to development procedures for complete software system. Rather, individual algorithms can be developed with\linebreak CbC~\cite{watson2016correctness}. 
With the block rules, the scalability is improved because refinement steps that are easy to prove can be combined into one block. This saves the application of refinement rules and their corresponding correctness proofs. 
With \textsc{ArchiCorC}~\cite{knuppel2020scaling}, we can even scale CbC to the development of correct component-based architectures. By composing components specified with required and provided interfaces, we support the creation of software architectures correct by construction.}

As soon as we scale \TraitCbC to real languages, we have the same application scenarios as approaches that already use post-hoc program verification. As argued by Damiani et al.~\cite{damiani2014verifying}, traits enable an incremental process of specifying and verifying software. Bettini et al.~\cite{bettini2010implementing} proposed to use traits for software product line development and highlighted the benefits of fine-grained reuse mechanisms. Here, \TraitCbC's guideline is suitable for constructing new product lines step by step from the beginning.

\hl{Since \CbCBlock extends classic CbC and can be freely applied at any granularity of refinement steps, we propose to use \CbCBlock for any implementation of correctness-critical software, but the CbC approach must be well understood by the developer to be efficiently usable.
In \TraitCbC, methods are developed and composed directly, so less knowledge is needed to apply the approach, but developers can fall back into a post-hoc verification process and thus lose the benefits of CbC (e.g., if the developers first develop all methods and do not directly prove the correctness). In general, both approaches are usable for program development and the right choice depends on the preferences and prior knowledge of the developers.}

\noindent\textbf{Summary.}
In summary, \TraitCbC and \CbCBlock allow more flexible program construction without losing the advantages of incremental correct-by-construction program development.  \CbCBlock loosens the strict guideline of classic CbC by adding the block refinement rules. \CbCBlock still needs specialized tools, such as \corc to be applicable.
\TraitCbC enables a CbC approach for trait-based languages without introducing refinement rules. This program construction approach combined with the flexibility of traits allows correct methods to be developed in small and reusable steps. \TraitCbC is independent of special CbC tools and requires only a program verifier.

\section{Related Work}

In the following, we discuss related work for specifying and verifying software. We discuss related correctness-by-construction approaches and compare \corc with other tools for CbC.

\paragraph{Contracts and Program Verification}

The implementation of CbC in \corc and the implementation of \TraitCbC use JML, Java DL and Java to specify and write programs. For the verification, KeY~\cite{ahrendt2016deductive} is integrated in the backend. KeY is a deductive program verifier for Java programs specified with JML. In an intermediate step, the specified programs are translated to Java DL. Similarly, OpenJML~\cite{cok2011openjml} verifies Java programs specified with JML.
Besides Java/JML, many languages support pre-/postcondition contracts or other forms of specification to state program behavior. First, the programming language Eiffel introduced contracts and supported the design-by-contract approach~\cite{meyer1988eiffel,meyer1992applying}. Eiffel is an object-oriented programming language, where classes are specified with invariants, and methods with pre-/postconditions contracts. For the verification, AutoProof~\cite{khazeev2016initial,tschannen2015autoproof} is integrated that translates the specified program to a logic formula. Then, an SMT-solver proves the validity of the formulas.
For C\#, the language Spec\#  is an extension to introduce contracts and invariants~\cite{barnett2004spec,Barnett:2011:SVS:1953122.1953145}. The verification is done by translating the proof obligations to an intermediate language BoogiePL that can be verified with Boogie~\cite{barnett2005boogie}.
For the C language, the VCC~\cite{cohen2009vcc} and Frama-C~\cite{cuoq2012frama} tools verify annotated C code. VCC reuses the Spec\# tool chain. For Java and C, the VeriFast~\cite{jacobs2010quick} tool verifies C and Java programs.
VerCors~\cite{amighi2014verification} also support the verification of C and Java programs with a focus on concurrent and distributed software.
Another language with integrated specifications and verification is Dafny~\cite{leino2010dafny}. Dafny is a functional language, but supports the compilation to other languages such as C\#, Java, Go, and Python. Similarly, Whiley~\cite{pearce2013whiley} is a designed language with associated verifier to simplify the verification of programs.
The languages SPARK~\cite{barnes2003high} supports a subset of the Ada language to specify and verify Ada programs. In contrast to JML, the specification is not written as comments, but the Ada \textit{aspect}-syntax is used to express contracts.
The focus of all these languages and verification tools is the specification of program behavior and the verification that a program satisfies its specification.
With CbC (\CbCBlock and \TraitCbC), we put the correct construction of programs in the foreground, instead of just verifying the correctness post-hoc. However, Watson et al.~\cite{watson2016correctness} argue that correctness-by-construction and post-hoc verification can be used together to combine their mutual strengths.

To verify trait languages, Damiani et al.~\cite{damiani2014verifying} added specifications of methods in traits to verify correct trait composition. They proposed a modular and incremental verification process.
Traits are introduced in many languages to support clean design and reuse, for example Smalltalk~\cite{ducasse2006traits}, Java~\cite{bono2014trait} by utilizing default methods in interfaces, and other Java-like languages~\cite{bettini2013traitrecordj,liquori2008feathertrait,smith2005chai}.
None of these trait languages were used to formulate a CbC approach to create correct programs. They only focus on code reuse or post-hoc verification.

\paragraph{Refinement-Based Correctness-by-Construction}

The main idea of correctness-by-construction is the stepwise construction of a program from a starting specification with correctness guarantees for each step. 
We focused on correctness-by-construction by Kourie and Watson~\cite{kourie2012correctness} that we called classic CbC. This classic CbC approach is based on Dijkstra~\cite{dijkstra-book} and Gries~\cite{gries-book}.
In this paragraph, we discuss related refinement-based CbC approaches. All of these approaches create correct programs by refining an abstract program or system to a concrete implementation. This is the main difference to the composition-based CbC approach of \TraitCbC, where atomic units of code are composed to whole programs.

%For Gries' program development method, Heisel~\cite{Heisel1992} propose tool support that builds correctness proofs of the program under development. This tool is tailored to the guarded command language in contrast to \corc that supports a subset of Java.
%Related rule-based CbC approaches are Morgan's refinement calculus~\cite{morgan-book} and invariant base programming~\cite{back2012refinement,back2009invariant}.
Morgan's refinement calculus~\cite{morgan-book} is similar to correctness-by-construction by Kourie and Watson~\cite{kourie2012correctness}. Both approaches have the same theoretical foundation, but Morgan's refinement calculus is more elaborated with a large number of different refinement rules, where many rules are only formally interesting. Kourie and Watson~\cite{kourie2012correctness} reduced the refinement rules to a minimal but sufficient set, such that CbC becomes comprehensible for developers without a major background in formal methods. The language ArcAngel~\cite{oliveira2003arcangel} with the verifier ProofPower~\cite{zeyda2009proofpower} implements Morgan's refinement calculus. The tool uses a tactic language to apply a sequence of refinement rules for program refinement. Thus, a tactic has the same benefit as the application of a block refinement in \CbCBlock because the application of refinement rules is condensed to one refinement step. The difference is that for an introduced block of code in \CbCBlock, it does not matter what classic CbC refinement rules would have to be applied to introduce that block of code. A tactic still applies the refinement rules stated in that tactic sequentially.

%A list of proof obligations are discharged to ProofPower to verify the correct application of the tactics.
The invariant based programming~\cite{back2012refinement,back2009invariant} shifts the focus from pre-/postcondition contracts as starting point for refinements to invariants. The tool SOCOS~\cite{back2009invariant,back2007testing} implements Back's methodology. Similar to \corc, SOCOS has a graphical user interface to create a program in the form of a UML-style state chart. Refinement steps introduce new states and transitions in the state chart and check compliance with the invariants. A completely refined program is proved correct and executable code can be generated. In \corc, the  graphical user interface present the refinement steps in a hierarchical tree structure that more directly represent the structure of the code (comparable with an abstract syntax tree). Therefore, \corc and also the implementation of \TraitCbC are on the level of source code.

Further refinement-based methodologies are Event-B~\cite{EventB,abrial2010rodin} for automata-based systems and Circus~\cite{oliveira2009utp,oliveira2008crefine} for state-rich reactive systems. 
Both methodologies work on an abstraction level with abstract models instead of specified source code. In refinement steps these abstract models of the system are transformed to concrete and executable implementation. Here, each refined result guarantees conformations with the initial model.
Event-B is supported by the tool Rodin~\cite{abrial2010rodin}, and Circus is supported by the tool CRefine~\cite{oliveira2008crefine}.
The main difference to CbC by Kourie and Watson~\cite{kourie2012correctness}, and \TraitCbC is the abstraction level. We specify and verify source code rather than automata-based systems.

\hl{Data refinement~\cite{Haftmann2013,haslbeck2022few,lammich2012applying,cohen2013refinements} is a related approach that focuses on the refinement of (abstract) programs with abstract types to correct and more efficient programs with concrete types. Haftmann et al.~\cite{Haftmann2013} examine how the Isabelle/HOL code generator applies data refinements to produce executable versions of abstract programs. Cohen et al. \cite{cohen2013refinements} present an approach to refine Coq programs to enhance computational efficiency.
Haslbeck and Lammich~\cite{haslbeck2022few} not only ensure functional correctness during data refinement, they also verify worst-case complexity of algorithms at the
LLVM level. The main difference to CbC by Kourie and Watson~\cite{kourie2012correctness} is that data refinement approaches start with algorithms on abstract data structures that are refined to more concrete data structures, whether CbC by Kourie and Watson focuses on the incremental development of the algorithm itself. Therefore, both approaches can used in concert to develop more efficient algorithm.}

\paragraph{Extensions to Correctness-by-Construction and \corc}

\corc has been extended in several directions to allow the structured program development for larger software systems and further application areas. 
With \textsc{ArchiCorC}~\cite{knuppel2020scaling}, we integrate the construction of correct software architectures. We bundle \corc programs into reusable software components. The components communicate via required and provided interfaces where \textsc{ArchiCorC} guarantees the compatibility between them.
With \textsc{VarCorC}~\cite{bordis2020varcorc2} software product lines are developed correct by construction. A software product lines is used to systematically construct a family of similar software programs instead of developing monolithic programs. \textsc{VarCorC} ensures the correctness of all possible software variants of the product line.
In addition to functional correctness, correctness-by-construction and \corc are extended to guarantee nonfunctional properties. As a first example, we introduced CbC refinement rules to ensure that programs~\cite{runge2020lattice,runge2022ifbcoo} follow an information flow policy which defines the allowed flow of information in a program.
In every refinement step, security and functional correctness of the program is guaranteed, such that insecure and incorrect programs are prohibited by construction.
The goal of these extensions is that program development in \corc is scalable and that CbC can be used for additional application areas. Orthogonally, this article focuses on improving the flexibility of developing programs correct by construction (e.g., by introducing the block refinement rules).

\paragraph{Program and Specification Synthesis}

Program synthesis is a technique that generates programs from user given specifications automatically.
Pioneers in this field are Manna et al.~\cite{Manna1980}. Gulwani et al.~\cite{gulwani2017program} give an overview of state-of-the-art program synthesis approaches.
For example, for Fortran, Stickel et al.~\cite{stickel1994deductive} deductively extract programs from user-given graphical specifications. They compose procedures from libraries to full implementations. Similarly, Gulwani et al.~\cite{gulwani2010component} synthesize programs by composing base components from a specified library.
Polikarpova et al.~\cite{Polikarpova2016} synthesize recursive programs from specifications by utilizing type information.
Similarly, synthesis of function summaries~\cite{hoare1971procedures,Chen2015syn,sery2012summaries} automatically generate pre-/postcondition specifications from programs to achieve modular verification and to improve verification time.
With CbC (classic CbC, \CbCBlock, or \TraitCbC), developers have the task to specify and create programs according to that specification.
Therefore, CbC is a program development approach where the developer determines the resulting program, while program synthesis generates one of possibly many programs that fulfills the specification. Contrary to this, the synthesis of a function summary generates one of possibly many specifications for a program. Synthesis has scalability limitations due to an enormous search space of programs/specifications and ambiguity of user intent.

\section{Conclusion}

In this article, we presented \CbCBlock and \TraitCbC two incremental program construction approaches that guide developers to implementations that are correct by construction.
\CbCBlock extends classic CbC with block refinement rules. These rules allow to condense the application of CbC refinement rules into one block refinement. Thus, \CbCBlock
increase flexibility in the development of programs because any sequence of statements can be introduced in a block, while still ensuring the correctness of that introduced block.
\TraitCbC uses method calls and trait composition instead of refinement rules to guarantee functional correctness.
We formalize the concept of a trait-based object-oriented language with a parametric specification language to allow a broader range of languages to adopt this concept. The main advantage of \TraitCbC is the simplicity of the refinement process that supports code reuse.
We compared classic CbC, \CbCBlock, and \TraitCbC qualitatively with regard to their programming constructs, tool support, and usability. 
\CbCBlock and \TraitCbC both relax the strict guideline of CbC without losing the benefits of a constructive program construction approach.

\hl{As future work, user studies could be conducted with all three approaches to further evaluate the usability of the approaches. We want to investigate how the more flexible construction approaches of \TraitCbC and \CbCBlock are received by developers. We also want to compare the development times and potential types of programming errors between the approaches. These user studies will help to develop concrete guidelines on which approach is appropriate under which circumstances and with which team.}

\paragraph*{Acknowledgments}
This work was partially supported by funding from the topic Engineering Secure Systems of the Helmholtz Association (HGF) and by KASTEL Security Research Labs (46.23.03).
We thank Frederik Fr\"oling for his work on \CbCBlock in his Master's Thesis.

\bibliographystyle{alphaurl}
\bibliography{bib_runge}

\appendix
\section{Interview Questions}
\label{appendix:questions}

\begin{enumerate}[align=left]
	\item Which task was more difficult and why?
	\item Which tasks were solved?
	\item What were the biggest problems during the development?
	\item Is the development according to CbC understandable?
	\item Is the use of the block rules understandable?
	\item Is the introduction of the block rules reasonable?
	\item Would you use the block rules when implementing according to CbC?
	\item How do you like the development in the textual editor?
	\item How do you like the development in the graphical editor?
	\item Is the textual or the graphical editor preferred?
	\item Which elements from the editors are particularly helpful or inadequate and why?
	\item What functionalities are still missing in the editors?
	\item What would it take for you to develop according to CbC in your workday?
\end{enumerate}

\end{document}